\pdfoutput=1 
\documentclass{vldb} 

\usepackage{graphicx}

\usepackage{times}

\usepackage{amsmath}

\usepackage{verbatim}
\usepackage{epsfig}
\usepackage[iso]{umlaute}
\usepackage{amssymb}
\usepackage{graphicx}
\usepackage{bbm}
\usepackage{subfigure}
\usepackage{ifthen}
\usepackage{url}
\usepackage{algorithm}
\usepackage{color}
\usepackage{marginnote}

\usepackage{algorithm}
\usepackage{algorithmic}
\usepackage{paralist}
\usepackage{tabularx}
\usepackage{ifthen}
\newcommand{\DB}{\ensuremath{\mathcal{D}}}
\newcommand{\db}{\ensuremath{\mathcal{D}}} %the database

\newcommand{\statevec}{\ensuremath{\vec{s}} }
\newcommand{\query}{\ensuremath{q} } %the query object
\newcommand{\sd}{\ensuremath{\mathcal{S}}} %spatial domain
\newcommand{\td}{\ensuremath{\mathcal{T}}} %time domain
\newcommand{\tint}{\ensuremath{t} }
\newcommand{\dbobj}{\ensuremath{o} }

\newcommand{\tm}{\ensuremath{M}} %transitionmatrix
\newcommand{\probm}{\ensuremath{J}}
\newcommand{\pw}{\ensuremath{\bullet}} %piecewise multiplication
\newtheorem{definition}{Definition}
\newtheorem{lemma}{Lemma}

\newtheorem{example}{Example}

\graphicspath{.}
\newboolean{TR}
\setboolean{TR}{false}

\title{Probabilistic Nearest Neighbor Queries on Uncertain Moving Object
Trajectories\\
\small{Submitted for Peer Review 01.05.2013}\\
	Please get the significantly revised camera-ready version under http://www.vldb.org/pvldb/vol7/p205-niedermayer.pdf}
\begin{document}

\numberofauthors{1}

\author{{ Johannes Niedermayer$^*$, Andreas Züfle$^*$, Tobias
Emrich$^*$,}\\ {Matthias
 Renz$^*$, Nikos Mamoulis$^o$, Lei Chen$^+$, Hans-Peter Kriegel$^*$}
 \vspace{1.2mm}\\
 \fontsize{10}{10}\selectfont\itshape 
 $~^{*}$Institute for Informatics, Ludwig-Maximilians-Universität~München\\
 \fontsize{9}{9}\selectfont\ttfamily\upshape
 \{niedermayer,zuefle,emrich,kriegel,renz\}@dbs.ifi.lmu.de\\
 \\
  \fontsize{10}{10}\selectfont\rmfamily\itshape
 $~^{o}$Department of Computer Science, University of Hong Kong\\
 \fontsize{9}{9}\selectfont\ttfamily\upshape nikos@cs.hku.hk
  \vspace{1.2mm}\\
 \fontsize{10}{10}\selectfont\rmfamily\itshape $~^{+}$Department of Computer
 Science and Engineering,
 Hong Kong University of Science and Technology\\
 \fontsize{9}{9}\selectfont\ttfamily\upshape leichen@cse.ust.hk
} \maketitle
\begin{abstract}
    Nearest neighbor
    (NN) queries in trajectory databases have received significant
    attention in the past, due to their application in spatio-temporal
    data analysis. Recent work has considered the realistic
    case where the trajectories are uncertain; however, only simple
    uncertainty models have been proposed, which do not allow for 
    accurate probabilistic search. In this paper, we fill this gap
    by addressing probabilistic nearest neighbor queries in databases with
    uncertain trajectories modeled by stochastic processes, specifically
    the Markov chain model. We study three nearest neighbor query semantics that take as input a
    query state or trajectory $q$ and a time interval.
    For some queries, we show that no polynomial time solution can
    be found. For problems that can be solved in PTIME, we present
    exact query evaluation algorithms, while for the general case, we
    propose a sophisticated sampling approach, which uses Bayesian inference to
    guarantee that sampled trajectories conform to the observation data stored in
    the database. This sampling approach can be used in Monte-Carlo
    based approximation solutions. We include an extensive experimental
    study to support our theoretical results.
\end{abstract}

\section{Introduction}
\label{sec:introduction} With the wide availability of satellite,
RFID, GPS, and sensor technologies, spatio-temporal data can be
collected in a massive scale. The efficient management of such
data is of great interest in a plethora of application domains:
from structural and environmental monitoring and weather
forecasting, through disaster/rescue management and remediation,
to Geographic Information Systems (GIS) and traffic control and
information systems. In most current research however, each acquired trajectory,
i.e., the function of a spatio-temporal object that maps each point in time to a position
in space, is assumed to be known entirely without any uncertainty. However, the
physical limitations of the sensing devices or limitations of the data collection process introduce sources of
uncertainty.

Specifically, it is usually not possible to {\em continuously} capture the
position of an object for each point of time. In an indoor tracking environment
where the movement of a person is captured using static RFID sensors, the
position of the people in-between two successive tracking events is also not
available; the same holds for geo-social applications such as FourSquare. Also,
the frequency of data collection is often decreased to save resources such as battery power and wireless network traffic. In such scenarios, for a given
moving object, only a limited set of (location, time) observations is available.
In-between these observations the exact values are not explicitly stored in the
database and are thus uncertain.

In this work, we consider a database $\db$ of uncertain moving object
trajectories, where for each trajectory there is a set of %(uncertain)
observations for only some of the history timestamps. Thus, the
entire trajectory of an object is described by a time-dependent random
variable, i.e., a stochastic process.
Given a reference state or trajectory $q$ and a time interval $T$, we define
{\em probabilistic} nearest-neighbor (PNN) query semantics, which are
extensions of nearest neighbor queries in trajectory databases
\cite{FreGraPelThe07,GueBehXu10,IweSamSmi03,TaoPapShe02}. Specifically, a
P$\exists$NN (P$\forall$NN) query retrieves all objects in $\db$, which have sufficiently high probability to
be the NN of $q$ at one time (at the entire set of times) in $T$; a
probabilistic {\em  continuous} NN (PCNN) query finds for each
object $o\in \db$ the time subsets $T_i$ of $T$, wherein $o$
has high enough probability to be the NN of $q$ at the entire set of
times in $T_i$.

PNN queries find several applications in analyzing historical trajectory data.
For example, consider a geo-social network where users can publish their current
spatial position at any time by so-called check-ins. For a historical event,
users might want to find their nearest friends during this event, e.g. to share
pictures and experiences. As another application example, consider a collection
$\db$ of uncertain animal movements. PNN queries can be used to analyze animal
behavior according to how they moved relatively to a reference animal or object
(e.g., identify the {\em defender} of a herd who moves closely to and repulses
an attacking animal).

The main contributions of our work are as follows:

\begin{compactitem}
\item A thorough theoretical complexity analysis for variants of
probabilistic NN query problems and the identification of such
problems that are computationally hard.

\item Efficient algorithms for the introduced PNN query problems
that can be solved in polynomial time.

\item A sampling-based approximate solution for classes of hard
PNN problems which is based on Bayesian inference.

\item Hierarchical pruning strategies to speed-up
PNN queries exploiting the UST tree index \cite{EmrKriMamRenZue12}.

\item Thorough experimental evaluation of the proposed
concepts on real and synthetic data.

\end{compactitem}

The rest of the paper is structured as follows.
Section \ref{sec:related} reviews existing work related to NN search on
uncertain trajectories. Section \ref{sec:problemdef} reviews the Markov model
used to describe uncertain trajectories. Variants of NN search semantics based
on this model are also formally presented in this section. Section
\ref{sec:query} gives a theoretical analysis for each variant, identifying computationally hard
variants, while Section \ref{sec:pnnalgorithms} sketches polynomial solutions
for the remaining ones. An approximate solution is presented in Section
\ref{sec:sampling} that can be efficiently applied to any NN variant, yielding
very accurate results. This approach is based on Bayesian inference which adapts the Markov model
by conditioning its probabilities to observations in the future. The adaption is
necessary in order to obtain sample trajectories that are guaranteed to conform
with all observations. In Section \ref{sec:pruning}, we present a hierarchical
pruning strategy for NN queries on uncertain trajectories, which utilizes an
index structure proposed in our previous work.
An experimental evaluation of the efficiency and effectiveness of the proposed
techniques is presented in Section \ref{sec:experiments}.
Section \ref{sec:knn} briefly discusses how the presented solutions can be
easily adapted for the case of $k$NN queries having $k>1$ and shows the
complexity of such queries.
Finally, Section \ref{sec:conclusions} concludes this work.

\vspace{-1EM}
\section{Related Work}\label{sec:related}
Within the last decade, a considerable amount of research effort has been
      put into query processing in trajectory databases
      (e.g. \cite{TaoPapShe02, TaoFalPapLiu04, YuPuKou05, YioMokAre05,
      GueBehXu10}).
       In these works, the trajectories have been assumed to be certain, by employing linear \cite{TaoPapShe02} or more complex \cite{TaoFalPapLiu04} types of
      interpolation to supplement sparse observational data. However,
employing linear
      interpolation between consecutive observations might create
      impossible patterns of movement, such as cars travelling through lakes
      or similar impossible-to-cross terrain.
      Furthermore, treating the data as uncertain and answering
      probabilistic queries over them offers better
      insights\footnote{http://infoblog.stanford.edu/2008/07/why-uncertainty-in-data-is-great-posted.html}.

      \textbf{Uncertain Trajectory Modeling.}
      Several models of uncertainty paired
      with appropriate query evaluation techniques have been
      proposed for moving object trajectories (e.g. \cite{MokSu04,TraWolHinCha04,TraTamDinSchCru09,EmrEtAl11}). Many of these
      techniques aim at providing conservative bounds for the positions of
      uncertain objects. This can be achieved by employing geometric objects
      such as cylinders \cite{TraWolHinCha04, TraTamDinSchCru09} or beads
      \cite{TraChoWolYeLi10} as trajectory approximations. While these
      approaches allow to answer queries such as ``is it possible for object $o$ to
      intersect a query window $q$'', they are not able to assign probabilities
      to these events.

      Other approaches use independent probability density functions (pdf) at each point of time
      to model the uncertain positions of an object
      \cite{CheKalPra04,TraTamDinSchetal09,MokSu04}. However, as shown in
      \cite{EmrEtAl11}, this may produce wrong results (not in accordance
      with possible world semantics) for queries referring to a time
      interval because they ignore the
      temporal dependence between consecutive object positions in
      time. To capture such dependencies, recent approaches model the uncertain movement of objects based on stochastic processes.
In particular, in \cite{QiaTanJinLonDaiKuCha10,
  EmrEtAl11,ReLetBalSuc08,XuGuCheQiaoYu13}, trajectories are modeled based on
  Markov chains. This approach permits correct consideration of possible world semantics in
  the trajectory domain.

  Similar to our algorithm for adapting transition matrices is the
Baum-Welch algorithm for hidden Markov models (HMMs). This
algorithm aims at estimating \textit{time-invariant} transition
matrices and emission probabilities of a hidden Markov model. In
contrast, we assume this underlying model to be given, however we
aim at adapting it by computing
\textit{time-variant} transition matrices. 
Related to our algorithm is also the Forward-Backward-Algorithm for HMMs that
aims at computing the \textit{state distribution} of an HMM for each
point in time. In contrast, we aim at computing transition
matrices for each point in time, given a set of observations.

      \textbf{Nearest Neighbor Queries in Trajectory Databases.}
      In the context of certain trajectory databases there is not a
      common definition of nearest neighbor queries, but rather a set of
      different interpretations. In \cite{FreGraPelThe07}, given a query
      trajectory (or spatial point) $q$ and a time interval $T$, a NN query
      returns either the trajectory from the database which is closest to $q$
      during $T$ or for each $t \in T$ the
      trajectory which is closest to $q$. The latter problem has also been addressed in \cite{GueBehXu10}. Similarly,
      in \cite{KolGunTso99}, all trajectories which are nearest
      neighbors to $q$ for at least one point of time $t$ are computed.

      Other approaches consider \textit{continuous} nearest neighbor (CNN)
      semantics. In \cite{IweSamSmi03}, CNN queries were defined taking as input
      a static spatial query point $q$ and a trajectory database and returning for
      each point in time the trajectory closest to $q$. Other approaches
      \cite{PraWolChaDao97, TaoPapShe02} define the CNN problem differently:
      Given an input trajectory $q$ and a database consisting of spatial points,
      a CNN query segments $q$ such that for each segment $q_i \subseteq q$
      exactly one object from the database is the nearest neighbor of $q_i$.
      This approach was extended for objects with uncertain velocity and
      direction (thus considering a predictive setting rather than historical
      data) in \cite{HuaLiaLee08}; the solutions proposed only find
      possible results, but not result probabilities. Solutions for
      road network data were also proposed for the case where the velocities of  objects are
      unknown \cite{LiLiShuFan11}. Furthermore, \cite{TraTamDinSchCru09,TraTamCruSchHarZam11}
      extended the problem of continuous $k$NN queries (on historical
      search) to an uncertain setting, serving as important preliminary
      work, however, based on a model which is not capable to return answers
      according to possible world semantics.
Still, to date, there is no
      previous work addressing probabilistic NN queries over trajectory
      databases, considering possible world semantics.

\section{Problem Definition}
\label{sec:problemdef}
A spatio-temporal database $\db$ stores triples ($\dbobj_i$, \emph{time},
\emph{location}), where $\dbobj_i$ is a unique object
identifier, \emph{time $\in \td$} is a point in time and \emph{location $\in
\sd$} is a position in space. Semantically, each such triple corresponds to an \emph{observation} that object
 $\dbobj_i$ has been seen at some \emph{location} at some \emph{time}. In $\DB$, an object $\dbobj_i$
 can be described by a function $\dbobj_i(t) : \td \rightarrow \sd$ that maps each
 point in time to a location in space; this function is called
 \emph{trajectory}.

In this work, we assume  a discrete time domain $\td = \{0, \ldots, n\}$.
Thus, a trajectory becomes a sequence, i.e., a function on a discrete and ordinal
scaled domain. Furthermore, we assume a discrete state space of possible locations
(\textit{states}): $\sd = \{s_1, ..., s_{|\sd|}\} \subset \mathbb{R}^d$, i.e., we
use a finite alphabet of possible locations in a $d$-dimensional space. The way
of discretizing space is application-dependent: for example, in traffic applications we may use road crossings, in indoor tracking
applications we may use the positions of RFID trackers and rooms, and for
free-space movement we may use a simple grid for discretization.

\subsection{Uncertain Trajectory Model}
\label{subsec:model}
Let $\DB=\{\dbobj_1,...,\dbobj_{|\DB|}\}$ be a database containing
the trajectories of $|\DB|$ uncertain moving objects. For each
object $\dbobj \in \DB$ we store a set of observations
$\Theta^\dbobj=\{\langle \tint_1^\dbobj,\theta_1^\dbobj \rangle,
\langle \tint_2^\dbobj, \theta_2^\dbobj \rangle, \ldots, \langle
\tint_{|\Theta^\dbobj|}^\dbobj,\theta_{|\Theta^\dbobj|}^\dbobj
\rangle  \}$ where $\tint_i^\dbobj \in \td$ denotes the time and
$\theta_i^\dbobj \in \sd$ the location of observation
$\Theta_i^o$. W.l.o.g. let $\tint_1^\dbobj <
\tint_2^\dbobj<\ldots<\tint_{|\Theta^\dbobj|}^\dbobj$. Note that
the location of an observation is assumed to be certain, while the
location of an object between two observations is uncertain.

According to \cite{EmrEtAl11}, we can interpret the location of an
uncertain moving object $\dbobj\in\DB$ at time $\tint$ as a
realization of a random variable $\dbobj(\tint)$. Given a time
interval $[\tint_s, \tint_e]$, the sequence of uncertain locations
of an object is a family of correlated random variables, i.e., a
stochastic process.

This definition allows us to assess the probability of a possible
trajectory, i.e., the realization of the corresponding stochastic
process. In this work we follow the approaches from
\cite{EmrEtAl11,EmrKriMamRenZue12,XuGuCheQiaoYu13} and employ the
first-order Markov chain model as a specific instance of a
stochastic process. The state space of the model is the spatial
domain $\sd$. State transitions are defined over the time domain
$\td$. In addition, the Markov chain model is based on the
assumption that the position $o(t + 1)$ of an uncertain object $o$
at time $t + 1$ only depends on the position $o(t)$ of $o$ at time
$t$. Please note that our proposed techniques can be easily
applied to Markov chains of arbitrary order without any further
adaption of the algorithms.

The conditional probability
$$
\tm_{ij}^\dbobj(t):=P(\dbobj(t+1)=s_j|\dbobj(t)=s_i)
$$
is the \emph{transition probability} of $\dbobj$
from state $s_i$ to state $s_j$ at a given time $t$. Transition probabilities
are stored in a matrix $\tm^\dbobj(t)$, called {\em transition matrix of object
o at time t}.
Let $\statevec^\dbobj(\tint)=(s_1,\ldots,s_{|\sd|})^T$ be the distribution
vector of an object $\dbobj$ at time $\tint$, where $\statevec_i^\dbobj(\tint) =
P(\dbobj(t) = s_i)$. Without any further knowledge (from observations) the
distribution vector $\statevec^\dbobj(\tint+1)$ can be inferred from
$\statevec^\dbobj(\tint)$:

$$\statevec^\dbobj(\tint+1)=\tm^\dbobj(\tint)^T \cdot \statevec^\dbobj(\tint)$$

The traditional Markov model \cite{EmrEtAl11} uses forward
probabilities only. In Section \ref{sec:sampling}, we propose a Bayesian inference
approach, to condition this \emph{a-priori Markov chain} to an
adapted \emph{a-posteriori Markov chain} which also considers all observations of an
object.

\subsection{Nearest Neighbor Queries}
\label{subsec:querydef}
In this work we consider three types of time-parameterized nearest-neighbor
queries that take as input a certain reference state or trajectory $q$ and a set
of timesteps $T$.

\begin{definition}[P$\exists$NN Query]
A probabilistic $\exists$ nearest neighbor query retrieves all
objects $o \in \DB$ which have a sufficiently high probability
($P\exists NN$) to be the nearest neighbor of $q$ for at least one
point of time $t \in T$, formally:
\begin{align*}
&P\exists NNQ(q,\db,T,\tau) = \{o\in\DB: P\exists NN(o,q,\DB,T)
\geq\tau\}\\
&\text{where } P\exists NN(o,q,\DB,T) =\\
&P(\exists t \in T: \forall o'
  \in \db \setminus o: d(q(t), o(t)) \leq d(q(t), o'(t)))
\end{align*}
and $d(x,y)$ is a distance function defined on
spatial points, typically the Euclidean distance.
\end{definition}

This definition is a straightforward extension of the
spatio-temporal query proposed in \cite{FreGraPelThe07}.
In addition, we consider NN queries with the $\forall$ quantifier
(introduced in  \cite{EmrEtAl11} for range queries).

\begin{definition}[P$\forall$NN Query]
A probabilistic $\forall$ nearest neighbor query retrieves all
objects $o\in\DB$ which have a sufficiently high probability
($P\forall NN$) to be the nearest neighbor of $q$ for the entire
set of timestamps $T$, formally:
\begin{align*}
&P\forall NNQ(q,\db,T,\tau) = \{o\in\DB: P\forall NN(o,q,\DB,T)
\geq\tau\}\\
&\text{where }  P\forall NN(o,q,\DB,T) = \\
&P(\forall t\in T: \forall o'  \in \db \setminus o:
  d(q(t),o(t))\leq d(q(t),o'(t)))
\end{align*}
\end{definition}

In addition to the $\exists$ and $\forall$ semantics for
probabilistic nearest neighbor queries we now introduce a
\emph{continuous} query type which intuitively extends the
spatio-temporal continuous nearest-neighbor query
\cite{PraWolChaDao97,
TaoPapShe02} to apply on uncertain trajectories. 

\begin{definition}[PCNN Query]
\label{def:pcnnq} A probabilistic continuous nearest neighbor
query retrieves all objects $o\in\DB$ together with the set
of timesets $\{T_i\}$ where in each $T_i$ the object has a
sufficiently high probability to be always the nearest neighbor of
$q(t)$, formally:
\begin{align*}
PCNNQ(&q,\db,T,\tau) = \\ &\{(o,T_i): o\in\DB, T_i\subseteq T, P\forall
NN(o,q,\DB,T_i)\geq\tau\}.
\end{align*}
Analogously to the CNN query definition \cite{PraWolChaDao97, TaoPapShe02}, in
order to reduce redundant answers it makes sense to redefine the PCNN Query where
we focus on results that maximize $|T_i|$, formally:
\begin{align*}
PCNNQ(q,\db,T,\tau) =& \\ \{(o,T_i): o\in\DB, &T_i\subseteq T,
P\forall NN(o,q,\DB,T_i)\geq\tau\ \\
&\wedge\forall T_j\supset T_i:
P\forall NN(o,q,\DB,T_i)<\tau\}.
\end{align*}

\end{definition}
Note that according to this definition result sets $T_i \subseteq T$ do not
have to be connected.

\subsection{Query Evaluation Framework and Roadmap}

An intuitive way to evaluate a PNN query is to compute for each $o
\in \DB$ the probability that $o$ is the NN of $q$ in at least one
or in all timestamps of a time interval $T$. However, to speed up
query evaluation, in Section \ref{sec:pruning}, we show that it is
possible to prune some objects from consideration using an index
over $\DB$. Then, for each remaining object $o$, we have to
compute a probability (i.e., $P\exists NN(o,q,\DB,T)$ or $P\forall
NN(o,q,\DB,T)$) and compare it to the threshold $\tau$. In Section
\ref{sec:query}, we show that while computing $P\exists
NN(o,q,\DB,T)$ is computationally hard,  $P\forall NN(o,q,\DB,T)$
can be computed in PTIME. Therefore, in Section
\ref{sec:pnnalgorithms}, we present an algorithm for computing
$P\forall NN(o,q,\DB,T)$ exactly and a technique that uses the
algorithm as a module for the computation of PCNN queries. As
discussed in Section \ref{sec:sampling}, the harder case of
retrieving $P\exists NN(o,q,\DB,T)$ (and also $P\forall
NN(o,q,\DB,T)$) can be approximated by Monte-Carlo simulation: for
each object $o'\in \DB$ a trajectory is generated which conforms
to both the Markov chain model $M^{o'}$ and the observations
$\Theta^{o'}$ and all these trajectories are used to model a
possible world. By performing the NN query in all these possible
worlds and averaging the results, we are able to derive an
approximate value for $P\exists NN(o,q,\DB,T)$.

\section{Theoretical Analysis}\label{sec:query}
In this section, we formally show that P$\exists$NN 
queries cannot be computed efficiently, in contrast to P$\forall$NN and 
PCNN queries.

\subsection{The P$\exists$NN Query} 
\label{sec:existsnn} In a  P$\exists$NN query, for any candidate object $o \in
\DB$, we should consider the probability $P\exists NN(o, q, \DB, T)$. However, the
following lemma shows that this probability is hard to compute.

\begin{lemma}
\label{lem:nphard}
The computation of $P\exists NN(o, q, \DB, T)$ is NP-hard in $|\DB|$.
\end{lemma}

\begin{proof}
$P\exists NN(o, q, \DB, T)$ is equal to $1-P(\neg\exists
  t \in T, \forall o' \in \DB:d(q(t),o(t)) \le d(q(t), o'(t)))$.
We will show that deciding if there exists a possible world for which the
expression:
  \begin{align}
  \label{eq:nphard}
    \neg\exists t \in T, \forall o' \in \DB:d(q(t),o(t)) \le d(q(t), o'(t))
  \end{align}
is satisfied is an NP-hard problem. (Note that this is a much easier problem
than computing the actual probability.) Specifically, we will reduce the well-known
  NP-hard $k$-SAT problem to the problem of deciding on the existence of a
  possible world for which Expression \ref{eq:nphard} holds.

For this purpose, we provide a mapping to convert a boolean formula
in conjunctive normal form to a Markov
  chain modeling the decision problem of Expression \ref{eq:nphard} in
  polynomial time. Thus, if the decision problem could be computed in PTIME,
then $k$-SAT  could also be solved in PTIME,
which would only be possible if P$=$NP.
  A $k$-SAT expression $E$ is based on a set of variables $X = \{x_1,
  x_2, \ldots, x_n\}$. The \textit{literal} $l_i$ of a variable $x_i$ is either $x_i$
  or $\neg x_i$ and a \textit{clause} $c = \bigvee\limits_{x_i \in \mathbb{C}} l_i$
  is a disjunction of literals where $\mathbb{C} \subseteq X$ and
  $|\mathbb{C}|<k$. Then $E$ is a conjunction of clauses: $E = c_1 \wedge c_2 \wedge \ldots \wedge c_m$.

  For our mapping, we will consider a simplified version of the P$\exists$NN
  problem, specifically (1) $q$ is a certain point, (2) $o$ is a
  certain point and (3) the state space $\sd$ of possible locations only includes 4 states. As
  illustrated in Figure \ref{fig:example}, compared to $o$, states $s_1$ and $s_2$
are closer to $q$ and states $s_3$ and $s_4$ are further
from $q$.\footnote{The states of $o$ and $q$
  are omitted for the sake of simplicity.} Therefore, if an uncertain
object is at states $s_1$ or $s_2$ then $o$ is not the NN of $q$.

  In our mapping, each variable $x_i \in X$ is equivalent to one uncertain
  object $o_i' \in \DB \setminus o$. Furthermore each disjunctive clause $c_j$
  is interpreted as an event happening at time $t=j$, i.e., the event $c_1$
  happens at time $t=1$,  $c_2$ happens at time $t = 2$ etc. Each clause
  $c_j$ can be seen as a disjunctive event that at least one object $o_i'$ at
  time $t=j$ is closer to $q$ than $o$ (in this case, $c_j$ is $true$).
  Therefore, the conjunction of all these events, i.e. expression $E =
  \mathop{\bigwedge}\limits_{1 \leq j \leq m} c_j$,
becomes true if the set of variables is chosen in a way that at each point in
  time, compared to $o$, at least one object is closer to $q$; this directly represents
  Expression \ref{eq:nphard}. However, in $k$-SAT, not every variable $x_i$
  (corresponding to $o_i'$) is contained in each term $c_j$ which does not
  correspond to our setting, since an uncertain object has to be
  \textit{somewhere} at each point in time.  To solve this problem, we extend
  each clause $c_j$, such that each variable $x_i$ is contained in $c_j$,
  without varying the semantics of $c_j$. Let us assume that $x_i$ is not
  contained in $c_j$. Then $c'_j = c_j \vee false = c_j \vee (x_i \wedge \neg
  x_i)$. This means that we can assume that object $o_i'$ is definitely not
  closer to $q$ than $o$ at time $t$.

  Let $l_i^j$ be the literal of variable $x_i$ in clause $c_j$. Based
  on the above discussion, we are able to construct for each object $o_i'$ two
  possible trajectories (worlds). The first one, based on the
  assumption that $x_i$ is true,
  transitions between states $s_2$ (if $l_i^j$ = true) and $s_4$ (if $l_i^j$ =
  false). The second one, based on the assumption that $x_i$ is set to false,
  transitions between states $s_1$ (if $l_i^j$ = true) and $s_3$ (if $l_i^j$ = false). Since these
  two trajectories can never be in the same state it is straightforward to
  construct a time-inhomogeneous Markov chain $\tm^o(t)$ for each object $o_i'$
  and each timestamp $j$.
  
  After the Markov chains for each uncertain object $o_i'$ in $\DB$ have been
  determined, we would just have to traverse them and compute the probability $P\exists
  NN(o, q, \DB, T)$. If this probability is $<$ 1, there would exist a solution to the
  corresponding $k$-SAT formula. However it is not possible to achieve this
  efficiently in the general case as long as $P \neq NP$. Therefore
  solving $P\exists NN$ in subexponential time is impossible. 
\end{proof}

\textbf{Example:}
Consider a set of boolean variables $X = \{x_1, \ldots, x_4\}$ and the following
formula: \[E = (\neg {x_1} \vee x_2 \vee x_3) \wedge (x_2 \vee
\neg {x_3} \vee x_4) \wedge (x_1 \vee \neg {x_2}) \] Therefore, we have
\[c_1 = (\neg {x_1} \vee x_2 \vee x_3), c_2 = (x_2 \vee \neg {x_3} \vee x_4)
\text{ and } c_3 = (x_1 \vee \neg{x_2})\] By employing the mapping discussed above,
we get the four inhomogeneous Markov chains illustrated in Figure
\ref{fig:example}. For instance, under the condition that $x_1$ is set to
$true$, the value of the literal $\neg {x_1}$ is false at $t =
1$ (in clause $c_1$) such that $o_1'$ starts in the state $s_4$. On the other
hand, if $x_1$ is set to $false$, then $o_1'$ starts in the state
$s_1$.

In the second clause $c_2$, since $x_1 \not \in \mathbb{C}_2$, the position of
$o_1'$ must not affect the result. Therefore, for both cases $x_1 = false$
and $x_1 = true$, $o_1'$ must be behind $o$. In the last clause $c_3$, if
$x_1 = true$ the object moves to state $s_2$. On the other hand, if $x_1 =
false$, the object moves to state $s_3$.

    \begin{figure}[t] \centering
    \includegraphics[width=0.8\columnwidth]{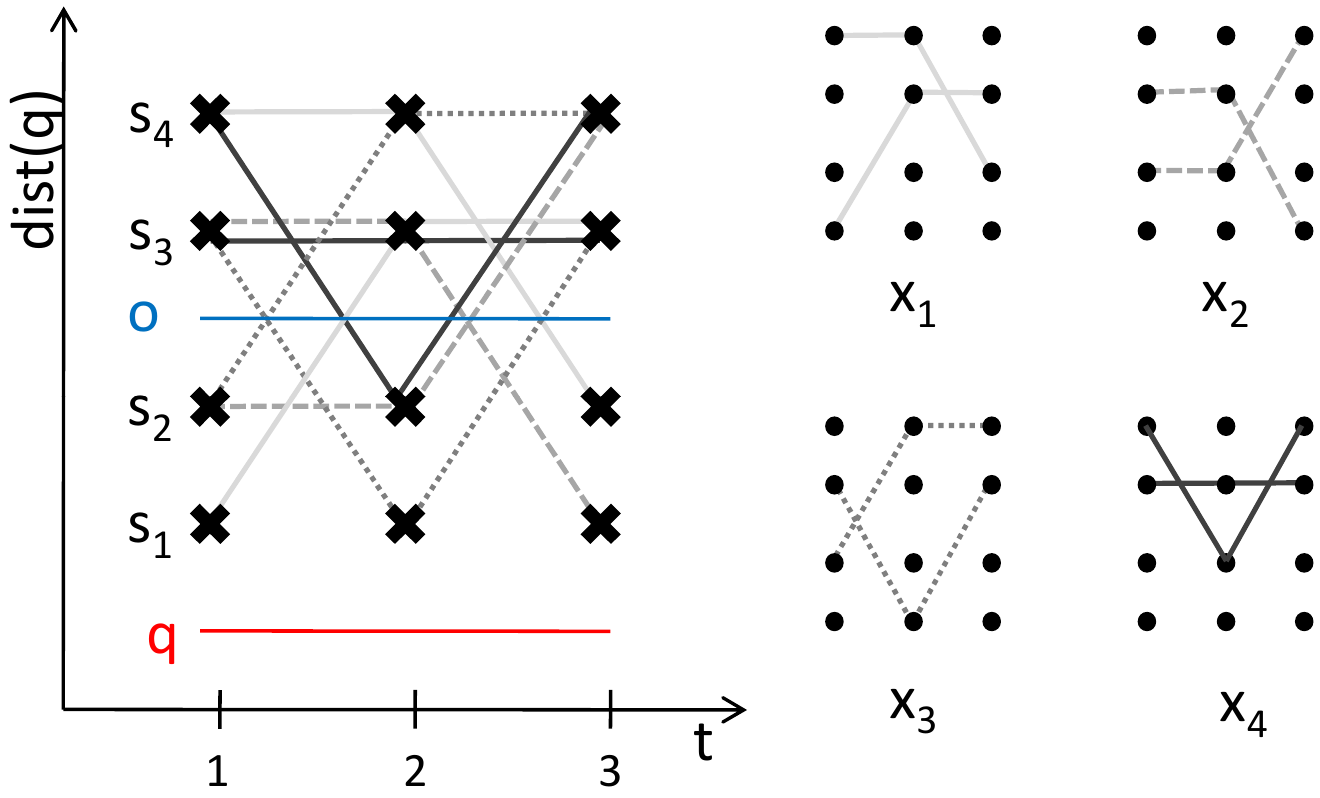}
      \caption{An example instance of our mapping of the 3-SAT problem to a set
      of Markov chains.}
      \label{fig:example}
      \vspace{-0.3cm}
    \end{figure}

\subsection{The P$\forall$NN Query}
Again we start our analysis by considering the single object
  probability $P\forall NN(o,q,\DB,T)$. The following lemma shows that
  this probability can be computed in PTIME.

  \begin{lemma}
  \label{lem:independent}
    The probability $P\forall NN(o,q,\DB,T)$ can be computed by considering
    each object $o' \in \DB$ independently, specifically $P\forall NN(o,q,\DB,T)
    = \prod\limits_{\dbobj' \in \DB} P\forall NN(o,q,\{o'\},T)$
  \end{lemma}
  \begin{proof}
 \begin{align*}
&    P\forall NN(o,q,\DB,T) &= \\
&    P(\forall \tint \in T, \forall o' \in \DB \setminus  o: d(q(t),o(t)) \le d(q(t),o'(t))) &= \\
&  P(\forall \tint \in T: d(q(t),o(t)) \le d(q(t),o'_1(t)) \wedge  \ldots &\\
&  \ldots \wedge d(q(t),o(t)) \le d(q(t),o'_{|\DB|}(t))  &= \\
&  \prod\limits_{\dbobj' \in \DB} P\forall NN(o,q,\{o'\},T) &
\end{align*}
  The last step follows from stochastic independence of
  objects and the resulting stochastic independence of events $A_i = \forall
  \tint \in T: d(q(t),o(t)) \le d(q(t),o'_i(t))$.
  \end{proof}

  Lemma \ref{lem:independent} allows us to further simplify the problem and we
  now only have to show how to compute $P\forall
  NN(o,q,\{o'\},T)$. This probability depends on two objects only and,
  thus, the computational complexity to compute $P\forall
  NN(o,q,\{o'\},T)$ is constant in $|\DB|$. In Section \ref{subsec:forallnnq} we provide an
  algorithm which computes $P\forall NN(o,q,\{o'\},T)$ in
  $O(|S|^3\cdot |T|)$ in the
  worst case, and in $O(|S|^2 \cdot  |T|)$ if the branching factor of the
  transition matrix is constant. This proves the polynomial complexity of the
  P$\forall$NN query.

\subsection{The PCNN Query}
The traditional CNN query \cite{PraWolChaDao97, TaoPapShe02},
retrieves the nearest neighbor of every point on a given query
trajectory in a time interval $T$. This basic definition usually returns $m <<
|T|)$ time intervals together having the same nearest neighbor. The main issue
when considering uncertain trajectories and extending the query definition
is the possibly large number of results due to highly overlapping and
alternating result intervals. In particular, considering Definition
\ref{def:pcnnq}, a PCNN query may produce an exponential number of results. This is because in the worst case each $T_i\subseteq T$ can be associated
with an object $o$ for which the probability $P\forall
NN(o,q,\DB,T_i)\geq\tau$, i.e., $2^T$ different $T_i$'s occur
in the result set.

To alleviate this issue, in Section \ref{subsec:PCNN_algorithm} we
propose a technique based on Apriori pattern mining to return the subsets of $T$
that have a probability greater than $\tau$.

\section{PNN Algorithms}
\label{sec:pnnalgorithms} In this section we present evaluation algorithms for
P$\forall$NN and PCNN queries.
Specifically, we focus on the computation of \linebreak $P\forall NN(o, q,
\DB, T)$ and $PCNN(o, q, \DB, T)$, respectively. For P$\exists$NN
queries no efficient algorithm exists as shown in Section
\ref{sec:existsnn}, thus an approximate (numerical) solution is
presented in Section \ref{sec:sampling}.

\subsection{The P$\forall$NN Query}
\label{subsec:forallnnq}

Note that Lemma \ref{lem:independent} allows us to compute $P\forall NN(o, q,
\DB, T)$ by considering the probabilities $P\forall NN(o, q,
\{o'\}, T)$ for each $o' \in \DB$ separately.

To derive an algorithm for
P$\forall$NN queries, we start by computing the probability that a candidate
object $o$ is the NN of query state or trajectory $q$ at a single point of time
$t$ in the query time window $T$. Let $\probm(\tint)$ be the joint probability matrix
of $o$ and $o'$, i.e., $\probm_{ij}(\tint) =$ \mbox{$P(o(\tint) = s_i \wedge o'(\tint)=s_j)$}, denoting the probability that object $o$ is in state $s_i$ at time $t$ \emph{and} object $o'$ is in state $s_j$ at time $\tint$.
  The matrix $\probm(\tint)$ can be computed by $\probm(\tint) = \statevec^{o}(\tint) \cdot
  \statevec^{o'}(\tint)^T$ due to independence of objects.
   From this joint probability matrix, we can derive the probability
  that $o$ is closer to $q$ than $o'$ at time $\tint$ as follows. We first
  define an indicator matrix $C(\tint)$ with
  \[C_{ij}(t) = \left\{\begin{array}{cl} 1, & \mbox{if } d(s_i, q(\tint)) \le
  d(s_j,q(\tint))\\
  0, & \mbox{otherwise} \end{array}\right.\]
  The matrix $C(\tint)$ describes for each state pair which state is closer to
  $q$. Therefore, we can determine the aggregated probability of $o$ being
  closer to $q$ than $o'$ by evaluating
  \begin{equation}
  \label{eq:hits_initial}
    H(\tint) = \probm(\tint)\pw C(\tint)
  \end{equation}
  where $\pw$ is the element-wise matrix multiplication. Then the following
  holds:
  \begin{equation}
  \label{eq:resultProb}
    P\forall NN(o,q,\{o'\},t) = \sum\limits_i\sum\limits_j(H_{ij}(\tint))
  \end{equation}
  This formula removes all possible worlds where $o'$ is closer to $q$ than $o$
  from the matrix $\probm(\tint)$ (Equality \ref{eq:hits_initial}) and sums up the
  remaining probabilities to get the result probability (Equality
  \ref{eq:resultProb}).
  To correctly incorporate later observations, we also have to consider all
  possible worlds not contributing to $P\forall NN(o,q,\{o'\},t)$. These possible worlds are stored in
  a drop matrix $D(t)$:
  \begin{equation}
  \label{eq:drops_initial}
  D(t) = \probm(\tint) - (\probm(\tint)\pw C(\tint))
  \end{equation}

  Based on the probability that $o$ is the NN of $q$ at a single time $t\in T$, we can
  compute the total probability that $o$ is the NN of $q$ during
  the complete time window $T =[t_s, t_e]$ by induction. The main
  idea is to maintain the matrix $H(t)$ over all times $t\in T =[t_s,
  t_e]$, which, contains, in each cell $H(t)_{ij}$ the probability
  that at time $t$, $o$ is located in state $s_i$ and $o'$ is
  located in state $s_j$ and $o$ has been closer to $q$ than $o'$
  during the whole interval $[t_s,t]$. Thus, the so-called
  \emph{Hit-Matrix} $H(t)_{ij}$ holds the probabilities of all
  possible worlds of $o$ and $o'$ where $o$ has always been closer
  to $q$ than $o'$. In addition to $H(t)_{ij}$, we further have to
  maintain a \emph{Drop-Matrix} $D(t)$, such that each cell $D(t)_{i,j}$ holds the
  probability that $o$ is located in state $s_i$ and $o'$ is
  located in state $s_j$ and $o$ was \emph{not} closer to $q$ than
  $o'$ at any time in $[t_s,t]$. These matrices partition all
  possible worlds into two classes: $H(t)$ represents all worlds still satisfying the
  query predicate at time $t$, and $D(t)$ contains all worlds that
  have already been pruned. Thus, for the first timestamp $\tint_s$, matrices
  $H(\tint_s)$ and $D(\tint_s)$) are computed; for the timestamps that follow,
  (1) $H(t_s)$ and $D(\tint_s)$) must be transitioned according to both
  transition matrices $\tm^{o}$ and $\tm^{o'}$ and (2) possible worlds have to
  be shifted into the correct matrix using $C(t)$. For each entry $H_{ij}(\tint)$ we have to aggregate over all states both objects $o$ and $o'$ can come from at the previous time step $\tint-1$:
  \begin{equation}
    H_{ij}(\tint) =
\sum\limits_k\sum\limits_l(H_{kl}(\tint-1)\cdot\tm^{o}_{ki}(\tint-1)\cdot\tm^{o'}_{lj}(\tint-1))
\cdot C_{i,j}(\tint),
  \end{equation}
  which is equal to:
\begin{equation}
\label{eq:exact_transition}
H^{tmp}(\tint) = [\tm^{o}(\tint-1)^T\cdot H(\tint-1)\cdot\tm^{o'}(\tint-1)]
\end{equation}
\begin{equation}
\label{eq:exact_transition_hit_to_drop}
H(\tint) = H^{tmp}(\tint)\pw C(\tint)
\end{equation}

A similar transformation has to be done for the \emph{Drop-Matrix}
$D(\tint)$ whereas shifting the possible worlds into $D(\tint)$
that have been hits in previous iterations but become drops at the
current time $\tint$, computed by $H^{tmp}(\tint)-H^{tmp}(\tint)\pw C(t)$:
\begin{equation}
\label{eq:exact_transition_d}
D^{tmp}(\tint) = \tm^{o}(\tint-1)^T\cdot D(\tint-1)\cdot\tm^{o'}(\tint-1)
\end{equation}
\begin{equation}
\label{eq:exact_transition_d2}
D(\tint) =D^{tmp}(\tint) + (H^{tmp}(\tint)-H^{tmp}(\tint)\pw C(t))
\end{equation}

  Until now, we did not consider observations for the computation of the result
  probability $P\forall NN(o,q,\{o'\},T)$. At each observation $\Theta^o_t$ and/or
  $\Theta^{o'}_t$, the joint probabilities $H(\tint)$ have to be {\em reweighed}
  according to the observation vectors $\theta^{o}_t$ and/or $\theta^{o'}_t$
  respectively. Specifically, given an observation $\Theta^o_x$ of object $o$
  at time $t = t_x^o$, the probabilities in $H(\tint)$
  have to be conditioned to the event of the observation ($o(t) = \theta^o_x$),
  which means $H(\tint)'$ ($H(\tint)$ after inferring the observation) should
  ultimately have the following form:
  $$
  H(\tint)_{ij}' = P(o(t) = s_i \wedge o'(t) = s_j \wedge H | o(t) = \theta^o_x)
  $$
  This expression denotes the probability that $o$ and $o'$ are in their
  respective states $s_i$ and $s_j$ and the set of possible worlds described by this
  probability being a \textit{hit} under the condition that $o$ was observed at
  time $t$ in state $\theta^o_x$. Clearly, the probability of $o$ being in state $s_i$ can only be non-zero if
  it is in state $\theta^o_x$. We can express this by introducing an indicator
  variable $I(s_i = \theta^o_x)$:
  $$
  H(\tint)_{ij}' = I(s_i = \theta^o_x) \cdot P(o(t) = s_i \wedge o'(t) = s_j \wedge
  H | o(t) = \theta^o_x)
  $$
  $$
  = I(s_i = \theta^o_x) \cdot P(o'(t) = s_j \wedge H | o(t) = \theta^o_x)
  $$
  By applying the law of conditional probability, we get:
  $$
  H(\tint)_{ij}' =\frac{I(s_i = \theta^o_x) \cdot P(o'(t) = s_j \wedge H \wedge o(t)
  = \theta^o_x)}{p(o(t)=s_i)}.
  $$
  While the nominator of this expression is already completely defined by $H(t)$
  and the observation $\theta^o_x$, for the denominator we also have to consider
  all possible worlds where $o$  is not a result of the $\forall$NN query:
  \begin{equation}
  \label{eq:exact_observation1}
   H(\tint)_{ij}' = \frac{I(s_i = \theta^o_x) \cdot H_{ij}}{\sum\limits_{X \in
   \{D,H\}}\sum\limits_{k=1}^{|S|} \sum\limits_{l=1}^{|S|} X_{kl}
   I(s_i = \theta^o_x)}
  \end{equation}
  Inferring an observation of
  object $o'$ ($o'(t) = \theta^{o'}_x$) can be derived similarly and results in
  the following formula:
  \begin{equation}
  \label{eq:exact_observation2}
   H(\tint)_{ij}' = \frac{I(s_j = \theta^{o'}_x) \cdot H_{ij}}{\sum\limits_{X
   \in \{D,H\}}\sum\limits_{k=0}^{|S|} \sum\limits_{l=0}^{|S|} X_{kl}
   I(s_j = \theta^{o'}_x)}
  \end{equation}
  Additionally, the same procedure has to be applied to $D(t)$. It can be
  further shown (by simple mathematical
  transformation) that both observations can be incorporated into $H(t)$ and
  $D(t)$ at the same time.

Algorithm \ref{alg:forallNN} is a pseudocode summarizing the findings of this
section. First, it computes the initial joint distribution
of $o$ and $o'$ (line \ref{alg_line:joint_distribution}) and the corresponding
hit and drop matrices (lines
\ref{alg_line:hits_initial}-\ref{alg_line:drops_initial}). Transitioning,
shifting and reweighting is performed for each point in time in lines
\ref{alg_line:transition_1} to \ref{alg_line:reweight_end}. The result probability is computed in line
\ref{alg_line:result_computation}.
The algorithm can be easily extended for the case where $T$ is a set
of disjoint intervals. In this case, for each point in time that does not
have to be considered as a query timestamp, the shifting operation (based on
matrix $C(t)$) needs not be evaluated. The remainder of the algorithm is not
affected by this extension. Using Algorithm \ref{alg:forallNN}, we are now able
to compute $P\forall NN(o,q,\{o'\},[t_s, t_e])$. For computing $P\forall
  NN(o,q,\DB,[t_s, t_e])$, an approach could apply this
  algorithm for all objects $o' \in \DB$.

  \begin{algorithm}
  \caption{$P\forall NN(o,q,\{o'\},[t_s, t_e])$}
  \begin{algorithmic}[1]
    \label{alg:forallNN}
        \STATE Generate $C(\tint_i)$ for all $t_i \in [t_s, t_e]$
        \STATE $\probm =
        \statevec^{o}(t_s)\cdot \statevec^{o'}(t_s)^{T}$
        \label{alg_line:joint_distribution}
        \STATE $H(\tint) = \probm(\tint)\pw C(\tint)$ \COMMENT{Eq.
        \ref{eq:hits_initial}}
        \label{alg_line:hits_initial} \STATE $D(t) = \probm(\tint) - (\probm(\tint)\pw
        C(\tint))$ \COMMENT{Eq. \ref{eq:drops_initial}}
        \label{alg_line:drops_initial}
        \FOR{$\tint = t_s + 1; t \le \tint_{e}; t++$}
          \STATE $H^{tmp}(\tint) = [\tm^{o}(\tint-1)^T \cdot  H(\tint-1) \cdot
          \tm^{o'}(\tint-1)]$ \COMMENT{Eq. \ref{eq:exact_transition}}
          \label{alg_line:transition_1}
          \STATE $H(\tint) = H^{tmp}(\tint)\pw C(t)$ \COMMENT{Eq.
          \ref{eq:exact_transition_hit_to_drop}}
          \STATE $D^{tmp}(\tint) = \tm^{o}(\tint-1)^T \cdot  D(\tint-1) \cdot
          \tm^{o'}(\tint-1)$ \COMMENT{Eq.
          \ref{eq:exact_transition_d}}
          \STATE $D(\tint) = D^{tmp}+(\tint)(H^{tmp}(\tint)-H^{tmp}(\tint)\pw
          C(t))$ \COMMENT{Eq.
          \ref{eq:exact_transition_d2}}
          \IF{$\exists \tint_i^{o}: \tint_i^{o} = \tint \vee \exists \tint_j^{o'}: \tint_j^{o'} = \tint$}
          \STATE reweigh according to Eq. \ref{eq:exact_observation1} or
          \ref{eq:exact_observation2}, respectively
          \ENDIF \label{alg_line:reweight_end}
        \ENDFOR
        \STATE $p = \sum\limits_i\sum\limits_j(H_{ij}(\tint_{e}))$ \COMMENT{Eq.
        \ref{eq:resultProb}} \label{alg_line:result_computation}
        \RETURN $p$
  \end{algorithmic}
\end{algorithm}

\subsection{The PCNN Query}
\label{subsec:PCNN_algorithm} Algorithm \ref{alg:PCNN} shows how
to compute, for a query trajectory $q$, a time interval $T$, a
probability threshold $\tau$, and an uncertain trajectory
$o\in\DB$ all $T_i \subseteq T$ for which $o$ is the nearest neighbor to
$q$ at all timestamps in $T_i$ with probability of at least $\tau$,
and the corresponding probabilities.

  \begin{algorithm}
  \caption{$PC_\tau NN$($q$, $o$, $\DB$, T, $\tau$)}
  \begin{algorithmic}[1]
    \label{alg:PCtauNN}
        \STATE $L_1 = \{(\{\tint\},P)| \tint\in T \wedge P\forall NN(o,q,\DB\setminus\{o\},
        \{\tint\}) \geq \tau\}$ \label{alg:PCNN_L1}
        \FOR{$k=2; L_{k-1}\not = \emptyset; k++$} \label{alg:PCNN_L2}
          \STATE $T^k = \{T_k\subseteq T| |T_k|=k\wedge \forall T'_{k-1} \subset
          T_k \exists (T'_{k-1}, p) \in L_{k-1}\}$ \label{alg:PCNN_L3}
          \STATE
          $L_k = \{(T_k,p) | T_k\in T^k \wedge P\forall NN(o,q,\DB\setminus\{o\},T_k) \ge \tau\}$ \label{alg:PCNN_L4}
        \ENDFOR
        \RETURN $\bigcup_k L_k$ \label{alg:PCNN_L7}
  \end{algorithmic}
  \label{alg:PCNN}
\end{algorithm}

We take advantage of the Apriori principle that for a $T_i$ to
qualify as a result of the PCNN query, all proper subsets of $T_i$ should
satisfy a P$\forall$NN query. In other words if $o$ is the P$\forall$NN of $q$
in $T_i$ with probability at least $\tau$, then for all $T_j\subset T_i$
$o$ should be the P$\forall$NN of $q$ in $T_j$ with
probability at least $\tau$.
Thus, we adapt the Apriori pattern-mining approach from
\cite{AgrSri94} to solve the problem as follows.
We start by computing the probabilities of all single points of time
to be query results (line~\ref{alg:PCNN_L1}).
Then, we iteratively consider the set $T^k$ of all timestamp sets with $k$
points of time by extending timestamp sets $T_{k-1}$ with an
additional point of time $t\in T\setminus T_{k-1}$, such that all $T'_{k-1}
\subset T_k$ have qualified at the previous iteration, i.e., we have $P\forall
NN(o,q,\DB\setminus\{o\},T'_{k-1}))\geq\tau$ (line~\ref{alg:PCNN_L3}).
The probability resulting from a
P$\forall$NN query is monotonically decreasing with the number of points in
time considered, i.e.,
P$\forall$NN($o,q,\db\setminus\{o\},T_k$) $\ge$
P$\forall$NN($\query,\db,T_{k+1}$) where $T_{k}\subset T_{k+1}$.
Therefore we do not have to further consider the set of points of time $T_k$
that do not qualify for the next iterations during the iterative construction of
sets of time points. Based on the sets of points in time $T_k$ constructed in
each iteration we compute the corresponding probability $P\forall
NN(o,q,\DB\setminus\{o\},T_k)$ to build the set of results of length $k$ (line~\ref{alg:PCNN_L4}) that are finally collected and reported as result
in line~\ref{alg:PCNN_L7}. The basic algorithm can be sped up by employing the
property that given P$\forall$NN($o,q,\db\setminus\{o\},T_1) =1$ the probability of P$\forall$NN($o,q,\db\setminus\{o\},T_1 \cup T_2) =$
P$\forall$NN($o,q,\db\setminus\{o\}, T_2)$.

Based on Algorithm \ref{alg:PCtauNN} it is possible to define a straightforward
algorithm for processing PCNN queries (by considering each object $o'$ from the
database). Again this approach can be improved by the use of an
appropriate index-structure (cf. Section \ref{sec:pruning}).

\vspace{-2mm} 
\section{Sampling Possible Trajectories}
\label{sec:sampling}

Based on the discussion in the previous sections, it is clear that
answering probabilistic queries over uncertain trajectory
databases has high run-time cost. Therefore,
like previous work \cite{JamXuWuetal08}, we study sampling-based
approximate solutions to improve query efficiency. In this
section, we first show that a traditional sampling approach is not
applicable for uncertain trajectory data as defined in Section
\ref{subsec:model}, as it does not account for all observations of
an object, resulting in a very large number of sample paths, which
are impossible given all observations. To tackle this issue, we
employ an approach that incorporates information about
observations directly into the Markov model, following a
forward-backward paradigm. Based on the resulting a-posteriori
models, traditional sampling approaches can be used to efficiently
and accurately estimate
   $PNN$ probabilities. On these samples, traditional NN
   algorithms for (certain) trajectories (\cite{FreGraPelThe07, GueBehXu10,
KolGunTso99, IweSamSmi03, PraWolChaDao97, TaoPapShe02}) can be
used to estimate NN probabilities.

\begin{figure}
\center
\vspace{-3mm}
\includegraphics[width=0.55\columnwidth]{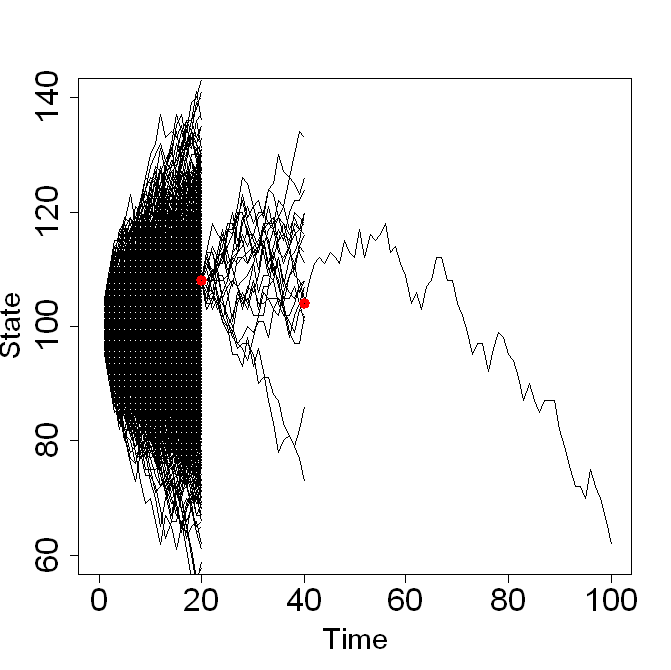}
\vspace{-4mm}
\caption{Traditional MC-Sampling.\label{fig:sampling}}
\vspace{-5mm}
\end{figure}

   \subsection{Traditional Sampling}\label{sec:sampling:naive}
   To sample possible trajectories of an object, a traditional
   Monte-Carlo approach would start by taking the first observation of
   the object, and then perform forward transitions using the 
   a-priori transition matrix. This approach however, cannot directly account for
   additional observations for latter timestamps. Figure
   \ref{fig:sampling} illustrates a total of 1000 samples drawn in a one-dimensional space. Starting at the first
   observation time $t=0$, transitions are performed using the
   a-priori Markov chain.
At the second observation at time $t=20$, the great majority of
   trajectories becomes inconsistent. Such impossible trajectories have
   to be dropped.
   At time $t=40$, even more trajectories become invalid;
   After this observation, only one out of a thousand samples
   remains possible and useful.

   Clearly, the number of trajectory generations required to obtain a
   single valid trajectory sample increases exponentially in the number of
   observations of an object, making this traditional Monte-Carlo
   approach inappropriate in obtaining a sufficient number of valid
   samples within acceptable time.

    \begin{figure}[t] \centering
    \includegraphics[width=\columnwidth]{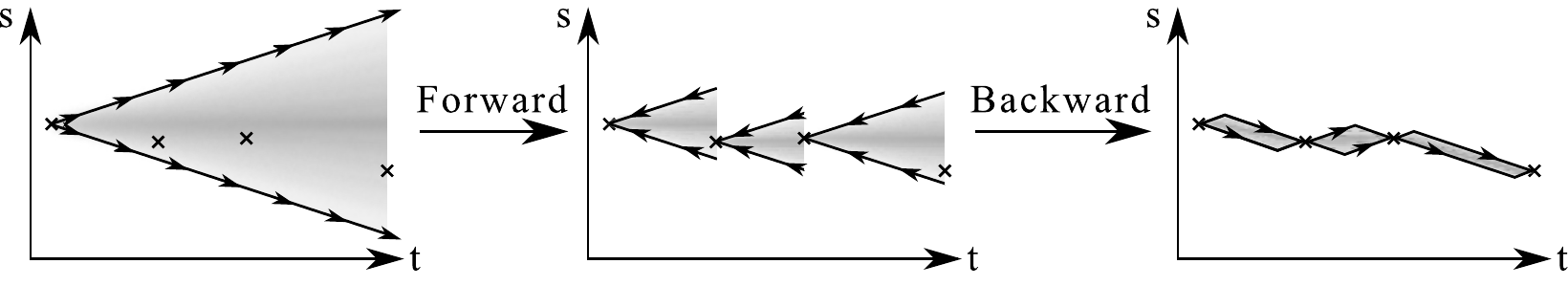}
      \vspace{-8mm}
      \caption{An overview over our forward-backward-algorithm.}
      \label{fig:fb_example}
      \vspace{-3mm}
    \end{figure}

\subsection{Efficient and Appropriate Sampling}
In a nutshell, our approach starts with the initial observation
$\theta^o_1$ at time $\tint_1^o$, and performs transitions for object
$o$ using the a-priori Markov chain of $o$ until the final
observation $\theta^o_{|\Theta^o|}$ at time $\tint_{|\Theta^o|}$
is reached. During this \emph{Forward}-run phase, Bayesian inference is
used to construct a time-reversed Markov-model $R^o(t)$ of $o$ at
time $t$ given observations in the past, i.e., a model that
describes the probability
$$
R^o_{ij}(t):=P(o(t-1)=s_j|o(t)=s_i,\{\theta^o_i|t^o_{i} < t\})
$$
of coming from a state $s_j$ at time $t-1$, given being at state
$s_i$ at time $t$ and the observations in the past.
Then, in a second \emph{Backward}-run phase, our approach traverses time
backwards, from time $\tint_{|\Theta^o|}$ to $\tint_1$, by employing
the time-reversed Markov-model $R^o(t)$ constructed in the forward
phase. Again, Bayesian inference is used to construct a new Markov
model $F^o(t-1)$ that is further adapted to incorporate knowledge
about observations in the future. This new Markov model contains
the transition probabilities
\begin{equation} \label{eq:m2}
F^o_{ij}(t-1):= P(o(t)=s_j|o(t-1)=s_i,\Theta^o).
\end{equation} for each point of time $t$, given all observations,
i.e., in the past, the present and the future.

As an illustration, Figure \ref{fig:fb_example}(a) shows the initial
model given by the a-priori Markov chain, using the first
observation only. In this case, a large set of ({\em time}, {\em location})
pairs can be reached with a probability greater than zero.
The adapted model after the forward phase (given by the a-priori Markov chain
and all observations), depicted in Figure \ref{fig:fb_example}(b), significantly
reduces the space of reachable ({\em time}, {\em location}) pairs and adapts respective
probabilities. The main goal of the forward-phase is to construct the necessary
data structures for efficient implementation of the
backward-phase, i.e.,
$R^o(t)$. This task is not trivial, since
the Markov property does
not hold for the future, i.e.,
the past is \emph{not}
conditionally independent of the future given the present.
Figure \ref{fig:fb_example}(c) shows the resulting model
after the backward phase. Next, both phases are elaborated in detail.

\subsubsection{Forward-Phase}
\label{sec:forward}

  To obtain the backward transition matrix $R^o(t)$, we can apply the theorem of Bayes as follows:

\begin{equation}\label{eq:bayes123}
R^o(t)_{ij} := P(o(t-1)=s_j|o(t)=s_i) =\end{equation}$$
\frac{P(o(t)=s_i|o(t-1)=s_j)\cdot P(o(t-1)=s_j)}{P(o(t)=s_i)}
$$
Computing $R^o(t)_{ij}$ is based on the
a-priori Markov chain only, and does not consider any information
provided by observations. To incorporate knowledge about past
observations into $R^o(t)_{ij}$,
let $past^o(t):=\{\theta^{\dbobj}_i|t^o_{i}< t\}$ denote the set
of observation temporally preceding $t$. Also, let
$prev^o(t):=argmax_{\Theta^{\dbobj}_i\in past^o(t)}t^o_{i}$
denote the most recent observation of $o$ at time $t$. Given all
past observations, Equation \ref{eq:bayes123} becomes conditioned
as follows:
\begin{lemma}
\label{lem:bayessimple}
\begin{equation}\label{eq:bayessimple}
R^o(t)_{ij} := P(o(t-1)=s_j|o(t)=s_i,past^o(t)) =\end{equation}$$
\frac{P(o(t)=s_i|o(t-1)=s_j,past^o(t)) \cdot
P(o(t-1)=s_j|past^o(t))}{P(o(t)=s_i|past^o(t))} $$
\end{lemma}
\begin{proof}
Equation \ref{eq:bayessimple} uses the conditional theorem of
Bayes\\ $P(A|B,C)=$ $\frac{P(B|A,C)\cdot P(A|C)}{P(B|C)}$, the
correctness of which is shown in the extended version of this
paper (\cite{NieZueEmrRenetal13}).
\end{proof}
 The conditional probability $P(o(t)=s_i|o(t-1)=s_j,past^o(t))$ can be rewritten as
 $P(o(t)=s_i|o(t-1)=s_j)$, exploiting the Markov property.

 Both priors $P(o(t-1)=s_j|past^o(t))$ and $P(o(t)=s_i|past^o(t))$ can
 be rewritten as $P(o(t-1)=s_j|prev^o(t))$ and
 $P(o(t)=s_i|prev^o(t))$ respectively, by exploiting the Markov property; i.e.,
 given the position at some time $t$, the
 position at a time $t^+>t$ is conditionally independent of the
 position at any time $t^-<t$.
Thus, Equation \ref{eq:bayessimple} can be rewritten as $R^o(t)_{ij}=$
\begin{equation}\label{eq:bayessimple2}
\frac{P(o(t)=s_i|o(t-1)=s_j) \cdot
P(o(t-1)=s_j|prev^o(t))}{P(o(t)=s_i|prev^o(t))}
\end{equation}
The probability $P(o(t)=s_i|o(t-1)=s_j)$ is given directly by the
definition of the  a-priori Markov chain $\tm^o(t)$ of $o$. Both
priors $P(o(t-1)=s_j|prev^o(t))$ and
 $P(o(t)=s_i|prev^o(t))$ can be
 computed by performing time transitions from observation
 $prev^o(t)$, also using the a-priori Markov chain $\tm^o(t)$. For each element $r_{ij}\in
 R^o(t)_{ij}$, and each point of time $t\in
 [\tint_{1},\tint_{|\Theta^o|}]$, these priors can be computed in
 a single run, iteratively performing transitions from $\tint_{1}$ to
 $\tint_{|\Theta^o|}$.
 During this run, all backward probabilities $P(o(t-1)=s_j|o(t)=s_i, past^o(t))$ are computed using Equation \ref{eq:bayessimple2} and memorized in the inhomogeneous matrix $R^o(t)$.
  During any iteration of the forward algorithm, where a new observation
$present^o(t):=\Theta^o_t\in\Theta^o$ is reached, the information
of this observation has to be incorporated into the model. This is
done trivially, by setting $P(o(t)=s_i|past^o(t),present^o(t))$ to one
if $s_i$ is the state $\theta$ observed by $present^o(t)$ and to
zero otherwise.

  \subsubsection{Backward Phase}
    During the backward phase, we traverse time backwards using the reverse transition matrix $R^o(t)$, to propagate
    information about future observations back to past points of time, as
    depicted in Figure \ref{fig:fb_example}(c). During this traversal,
    we again obtain a time reversed matrix $F^o(t)$,
    describing state transitions between adjacent points of time,
    given observations in the future.
Due to this second reversal of time, matrix $F^o(t)$ also contains
    adapted transition probabilities in the forward direction of
    time. 
    Thus, matrix $F^o(t)$ represents a Markov model which corresponds to the
    desired a-posteriori model: It contains the probabilities of performing a state transition between state $s_i$ and $s_j$ at time $t$ to time $t+1$, incorporating knowledge of observations
   in both the past and
    the future. In contrast, the a-priori Markov model $M^o(t)$ only
    considers past observations. We now discuss the details of this phase.

 By definition of $R^\dbobj(t)$ as the reverse transition matrix, the following
 reverse Markov property holds for each element $R^\dbobj_{ij}$ of $R^\dbobj$:
 $$
P(\dbobj(\tint)=s_j|\dbobj(\tint+1)=s_i,\dbobj(\tint+2)=s_{\tint+2},...,\dbobj(\tint+k)=s_{\tint+k})=
$$
\begin{equation} \label{eq:reversemarkov}
P(\dbobj(\tint)=s_j|\dbobj(\tint+1)=s_i)
\end{equation}

    As an initial state for the backward phase,
    we use the state vector
    corresponding to the final observation $\Theta^o_{|\Theta^o|}$ at time $t^o_{|\Theta^o|}$ at state $\theta^o_{|\Theta^o|}$.
    This way, we take the final observation as given, making any further probabilities that
    are being computed conditioned to this observation.
    At each point of time $t\in[\tint_{|\Theta^o|},\tint_1]$ and each state $s_i\in S$, we compute
    the probability that $o$ is located at state $s_i$ at time $t$
    \emph{given} (conditioned to the event) that the observations
    $future^o(t):=\{\theta^\dbobj_i|t^o_{i}>t)\}$ at times later than $t$ are
    made. In the following, let $next^o(t)=argmin_{\Theta^o_i\in
    future^o(t)}(t^o_i)$ denote the soonest observation of $o$
    after time $t$.
 To obtain $F^o(t)$, we once again
    exploit the theorem of Bayes:
    $$
F^o_{ij}(t) := P(o(t+1)=s_j|o(t)=s_i,\Theta^o)=$$
\begin{equation}
\label{eq:forward_matrix}
    \frac{P(o(t)=s_i|o(t+1)=s_j,\Theta^o)\cdot
    P(o(t+1)=s_j|\Theta^o)}{P(o(t)=s_i|\Theta^o)}
\end{equation}
    By exploiting the reverse Markov property (c.f. Equation
    \ref{eq:reversemarkov}), we can rewrite $P(o(t)=s_i|o(t+1)=s_j,\Theta^o) =$ $P(o(t)=s_i|o(t+1)=s_j, past(t+1))$
which is given by matrix $R^\dbobj(t)$. Both priors
$P(o(t+1)=s_j|\Theta^o)$ and $P(o(t)=s_i|\Theta^o)$ can be
rewritten as $P(o(t+1)=s_j|prev^o(t+1), present^o(t+1), next^o(t+1))$ and
$P(o(t)=s_i|prev^o(t), present^o(t), next^o(t))$, exploiting the traditional
Markov property in forward and Equation \ref{eq:reversemarkov} in backward
direction. These probabilities can be computed
as follows: We start at $t=\tint_{|\Theta^o|}$, performing transitions
backwards using backward transition matrix $R^o(t)$ until time
$t=\tint_{|\Theta^o|-1}$ is reached. For each intermediate point
of time $t$, the distribution vector $\statevec^\dbobj(\tint)$ is obtained.
Each probability $\statevec^\dbobj(\tint)_i$ in this vector corresponds
to the probability of $o$ being located at state $s_i$ at time $t$. These probabilities
are conditioned to $\Theta^o_{|\Theta|}=next^o(t)$, due to being
started according to $\Theta^o_{|\Theta|}$. Furthermore, these
probabilities are conditioned to $prev^o(t)\in past^o(t)$ due to
usage of matrix $R^o(t)$. At time $t^o_{|\Theta|-1}$, the state
vector is adapted using this observation. This procedure is
iterated until the first observation $\Theta^o_1$ is reached to derive the
probabilities $P(o(t+1)=s_j|\Theta^o)$ and $P(o(t)=s_i|\Theta^o)$.

\vspace{2mm}
\subsubsection{Sampling Process}
  \begin{algorithm}
  \caption{AdaptTransitionMatrices($\dbobj$)}
\small
  \begin{algorithmic}[1]
  \label{alg:transitionMatrices}
        \STATE \COMMENT{Forward-Phase}
        \STATE $\statevec^\dbobj(\tint^o_1) = \theta^o_1$
        \FOR{$\tint=\tint^o_1+1$; $\tint \le \tint^o_{|\Theta^o|}$;
        $\tint\mbox{++}$} \STATE $X'(t) = \tm^\dbobj(\tint-1)^T\cdot diag(\statevec^\dbobj(\tint-1))$
            \label{alg1:forward_mult}
            \STATE $\forall i \in \{1\ldots |\sd|\}: \statevec^\dbobj(\tint)_i =
            \sum\limits_{j=1}^{|S|}X'_{ij}(t)$ \label{alg1:forward_rowsum}
            \STATE $\forall i,j \in \{1\ldots |\sd|\}: R^\dbobj(\tint)_{ij} =
            \frac{X'_{ij}(t)}{\statevec^\dbobj(\tint)_i}$
            \label{alg1:forward_R}
            \IF {$t\in\Theta^o$}
\STATE $\statevec^\dbobj(t) =
            \theta^o_t$ \COMMENT{Incorporate observation} \label{alg1:observation}
\ENDIF
        \ENDFOR
         \STATE
         \COMMENT{Backward-Phase}
         \FOR{$\tint=\tint^o_{|\Theta^o|}-1; \tint \ge \tint^o_1;
         \tint\mbox{\hspace{0.02cm}-\hspace{0.02cm}-}$} \label{alg1:back_start}
         \STATE $X'(t) = R^\dbobj(\tint+1)^T\cdot diag(\statevec^\dbobj(\tint+1))$
         \STATE $\forall i \in \{1\ldots |\sd|\}: \statevec^\dbobj(\tint)_i =
         \sum\limits_{j=1}^{|S|}X'_{ij}(t)$ \label{alg1:backward_rowsum}
         \STATE $\forall i,j \in \{1\ldots |\sd|\}: F^\dbobj(\tint)_{ij} =
            \frac{X'_{ij}(t)}{\statevec^\dbobj(\tint)_i}$\label{alg1:back_end}
         \ENDFOR
\RETURN $F^o$
\end{algorithmic}
 \end{algorithm}
Algorithm
  \ref{alg:transitionMatrices} summarizes the construction of the
  transition model for a given object $o$.
  In the forward phase, the new distribution vector $\statevec^\dbobj(t)$ of $o$ at time $t$ and backward probability matrix $R^\dbobj(t)$ at time $t$ can be efficiently derived from the
  temporary matrix $X'(t)$, computed in Line \ref{alg1:forward_mult}.
  The equation is equivalent to a simple transition at time $t$, except that the state vector
  is converted to a diagonal matrix first. This trick allows to
  obtain a matrix describing the joint distribution of the position of $o$ at time
  $t-1$ and $t$. Formally, each entry $X'(t)_{i,j}$ corresponds to the
  probability $P(o(t-1)=s_j \wedge o(t)=s_i | past^\dbobj(t))$ which is
  equivalent to the \textit{numerator} of Equation
  \ref{eq:bayessimple}.\footnote{The proof for this transformation $P(A\cap B
  | C) = P(A|C)\cdot P(B|A,C)$ can be derived analogously to Lemma
  \ref{lem:bayessimple}.} To obtain the denominator of Eq.
  \ref{eq:bayessimple} we first compute the row-wise sum of $X'(t)$ in Line
  \ref{alg1:forward_rowsum}.
  The resulting vector directly corresponds to $\statevec^\dbobj(t)$,
  since for any matrix $A$ and vector $x$ it holds that $A\cdot x = rowsum(A
  \cdot diag(x))$. By employing this rowsum operation, only one matrix
  multiplication is required for computing $R^\dbobj(t)$ and
  $\statevec^\dbobj(t)$.

  Next, the elements of the temporary matrix $X'(t)$ and the elements of
  $\dbobj.\statevec(t)$ are normalized in Equation
  \ref{eq:bayessimple}, as shown in Line \ref{alg1:forward_R} of the
  algorithm.
  
  Finally, possible observations at time $t$ are integrated in Line
  \ref{alg1:observation}.
  In Lines \ref{alg1:back_start} to \ref{alg1:back_end}, the same
  procedure is followed in time-reversed direction, using the
  backward transition matrix $R^o(t)$ to compute the a-posteriori
  matrix $F^o(t)$.

  The overall complexity of this algorithm is $O(|T|\cdot|\sd|^2)$. The initial
  matrix multiplication requires $|\sd|^2$ multiplications. While the complexity of a
  matrix multiplication is in $O(|\sd|^3)$, the multiplication of a
  matrix with a diagonal matrix, i.e., $\tm^T\cdot s$ can be rewritten as
  $\tm^T_i\cdot s_{ii}$, which is actually a multiplication of a vector with a
  scalar, resulting in an overall complexity of $O(|\sd|^2)$. Re-diagonalization
  needs $|\sd|^2$ additions as well, such as re-normalizing the transition matrix,
  yielding $3 \cdot |T| \cdot|\sd|^2$ for the forward phase. The backward
  phase has the same complexity as the forward phase, leading to an overall
  complexity of $O(|T|\cdot|\sd|^2)$.

Once the transition matrices $F^o(t)$ for each point of time $t$ have been
computed, the
actual sampling process is simple: For each object $o$, each
sampling iteration starts at the initial position $\theta^o_1$ at
time $t^o_1$. Then, random transitions are performed, using
$F^o(t)$ until the final observation of
$o$ is reached. Doing this for each object $o\in\DB$, yields a
(certain) trajectory database, on which exact NN-queries can be
answered using previous work.
Since the event that an object $o$ is a $\forall$-NN ($\exists$-NN) of $q$ is a
binomial distributed random variable, we can use methods from statistics, such
as the Hoeffding's inequality (\cite{Hoe63}) to give a bound of the
estimation error, for a given number of samples.

\vspace{-0.2cm}
\section{Spatial Pruning}
\label{sec:pruning} Pruning objects in probabilistic NN search can
be achieved by employing appropriate index structures for querying
uncertain spatio-temporal data. In this work, we use the
\emph{UST-tree} \cite{EmrKriMamRenZue12}. In this section, we briefly
summarize the index and show how it can be employed to efficiently
prune irrelevant database objects, identify result candidates, and
find influence objects that might affect the $\forall$NN probability
of a candidate object.

\textbf{The UST-Tree. } Given an uncertain spatio-temporal object $o$, the main
idea of the UST-tree is to conservatively approximate the set of possible
({\em location}, {\em time}) pairs that $o$ could have possibly visited, given its
observations $\Theta^o$.
In a first
approximation step, these ({\em location}, {\em time}) pairs, as well as the
possible ({\em location}, {\em time}) pairs defined by $\Theta^o_{i}$ and
$\Theta^o_{i+1}$ are minimally bounded by rectangles. Such a
rectangle, for observations $\Theta^o_{i}$ and $\Theta^o_{i+1}$ is
defined by the time interval $[t^o_i,t^o_{i+1}]$, as well as the
minimal and maximal longitude and latitude values of all reachable
states. \begin{example} Consider Figure \ref{fig:example_pruning},
where four objects objects $A$, $B$, $C$ and $D$ are given by
three observations at time $0$, $5$ and $10$. For each object, the
set of possible states in the corresponding time intervals $[0,5]$
and $[5,10]$ is approximated by two minimum bounding rectangles.
For illustration, the set of possible states at each point
of time is also depicted by dashed rectangles.
\end{example}
The UST-tree indexes the resulting rectangles using an $R^*$-tree
(\cite{BecKriSchSee90}). We now discuss how such an index structure can
be used for the evaluation of P$\forall$NN and P$\exists$NN queries.

\textbf{Pruning candidates of P$\forall$NN queries.}
For a P$\forall$NN query, an object must have a non-zero probability
of being the closest object to $\query$, for each timestamp falling
into the query interval. As a consequence, to find candidate objects
for the P$\forall$NN query, we have to consider for all objects
$o\in \db$ whether for each $\tint\in \query.T$ there does not exist
another object $o'\in \db$ such that
$d_{min}(o(t),\query(t))>d_{max}(o'(t),\query(t))$. Here,
$d_{min}(o(t),\query(t))$ ($d_{max}(o(t),\query(t))$) denotes the
minimum (maximum) distance between the possible states of $o(t)$ and
$\query(t)$. Thus, the set of candidates $C_{\forall}(q)$ of a
P$\forall$NN is defined as:
\begin{gather*}
C_{\forall}(\query)=\{o\in\DB|\forall t\in \query.T:
d_{min}(o,\query) \le min_{o'\in\DB}d_{max}(o',\query)\}
\end{gather*}

Applying spatial pruning on
the leaf level of the UST-tree, we have to apply the $d_{min}$ and
$d_{max}$ distance computations on the minimum bounding rectangles
on the leaf level in consideration of the time intervals
associated with these leaf entries. In our example, given the
query point $\query$ with $\query.T=[2,8]$, only object $A$ is a
candidate, since $d_{min}(q(t),A(t))\leq d_{max}(q(t),o(t))$ for
all $o\in\DB$ in the time intervals [0,5] and [5,10], both
together covering $\query.T$. Objects $B$, $C$ and $D$ can be
safely pruned.

It is important to note that pruned objects, i.e., objects not
contained in $C_{\forall}(\query)$ may still affect the
$\forall$NN probability of other objects and even may prune other
objects. For example, though object $B$ is not a candidate, it
affects the $\forall$NN probability of all other objects and
contributes to prune possible worlds of object $A$, because
$d_{min}(\query(t),A(t))<d_{max}(\query(t),B(t))$ $\forall
t\in[5,10]$. All objects having at at least one timestamp
$\tint\in\query.T$ a non-zero probability being the NN of $\query$
may influence the $\forall$NN probability of other objects. Since
we need these objects for the verification step of both the exact
and the sampling algorithms, we have to maintain them in an
additional list $I_{\forall}(\query)=$
\begin{gather*}
\{o\in\DB|\exists t\in T: d_{min}(o(t),
\query(t)) \le min_{o'\in\DB}d_{max}(o'(t),\query(t))\}
\end{gather*}

To perform spatial pruning at the non-leaf level of the UST-tree, we
can analogously apply $d_{min}$ and $d_{max}$ on the MBRs of the
non-leaf level.

    \begin{figure}[t] \centering
    \includegraphics[width=\columnwidth]{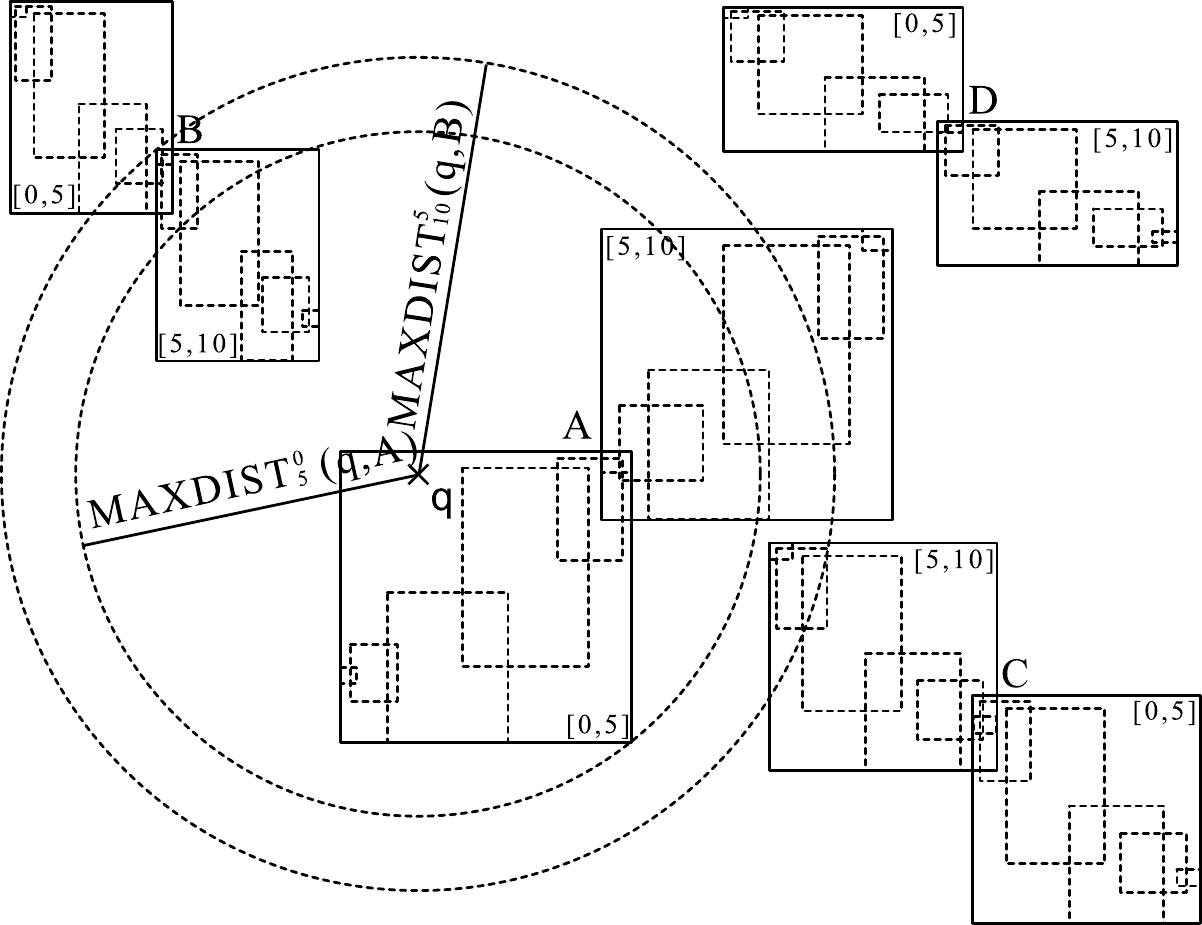}
     \vspace{-0.2cm}
      \caption{Spatio-Temporal Pruning Example. }\vspace{-0.2cm}
      \label{fig:example_pruning}

    \end{figure}

\textbf{Pruning for the P$\exists$NN query.} Pruning for the
P$\exists$NN query is very similar to that for the P$\forall$NN
query. However, we have to consider that an object being the nearest
neighbor for a single point in time is already a valid query result.
Therefore, no distinction is made between \textit{candidate}s and
\textit{influence object}s. Every pruner can be a valid result of
the P$\exists$NN query, such that each object with a $d_{min}$
smaller than the pruning distance has to be refined. The remaining
procedure of the P$\exists$NN-algorithm is equivalent to
P$\forall$NN-pruning.

\section{Experimental Evaluation}\label{sec:experiments}
  \textbf{Setup } Our experimental evaluation focuses on P$\forall$NN and PCNN
  queries, which have an efficient exact
  solution.

We conducted a set of experiments to verify both the effectiveness and efficiency of
the proposed solutions,
using
  a desktop computer having an Intel i7-870 CPU at 2.93 GHz and 8GB of RAM. All
  algorithms were implemented in C++ and integrated into the UST framework. The
  framework and a video illustrating the datasets can be found on the
  project page \cite{ProjectPage13}.
 
  \textbf{Artificial Data.} Artificial data for our experiments was created in
  three steps: state space generation, transition matrix construction and object
  creation. First, the data generator constructs a two-dimensional Euclidean
  \textit{state space}, consisting of $n$ states. Each of these states is drawn
  uniformly from the $[0,1]^2$ square. Then, in order to construct a transition
  matrix, we derive a graph by introducing edges between any point $p$ and its
  neighbors having a distance less than $r=\sqrt{\frac{b}{n*\pi}}$ with $b$
  denoting the average branching factor of the underlying network.
  This parameter ensures that the degree of a node does not
  depend on the number of states in the network. Each edge in the
  resulting network represents a non-zero entry in the transition matrix. We
  then set the transition probability of this entry indirectly proportional to
  the distance between the two connected vertices.

  \textbf{Real Data.} We also generate a data set is generated from a set
  of GPS trajectories of taxis in the city of Beijing \cite{YuaZheXieSun11} and has
  been provided by \cite{EmrKriMamRenZue12}.
  The data set was generated using the techniques of \cite{CheSheZhou11} to
  obtain both a set of possible states (corresponding to crossroads) and a
  transition matrix reflecting the possible movements of the cabs. This process yields
a state space consisting of about 3000 states and
  the corresponding transition matrices and direct edges between states.
  We assume that a-priori, all objects utilize the same Markov
  model $M$.

 \textbf{Observation Data.}
 To create observations of an object $o$, we sample a sequence of states and
 compute the shortest paths between them, modeling the motion of $o$ during
 its whole lifetime (which we set to 100 steps by default). To add uncertainty
 to the resulting path, every $l^{th}$ node, $l=i*v$, $v \in [0,1]$, of this
 trajectory is used as an observed state. $i$ denotes the
 time between consecutive observations and $v$ denotes a lag parameter
 describing the extra time that $o$ requires due to deviation from the shortest path; the
 smaller $v$, the more lag is introduced to $o$'s motion.
The resulting uncertain trajectories were distributed over the database time
horizon (default: 1000 timestamps) and indexed by a UST-tree \cite{EmrKriMamRenZue12}. As a
pruning step for query evaluation, we employed the UST-tree's MBR filtering
approach described in Section \ref{sec:pruning}. Query states were uniformly
drawn from the underlying state space.

  \subsection{Evaluation: P$\forall$NN Queries}
    For performance analysis, the sampling approach (Section \ref{sec:sampling}) is divided into
    two phases. In the first phase the trajectory sampler (\emph{TS}) is
    initialized (the adapted transition matrices are computed according to
    Algorithm \ref{alg:transitionMatrices}). This phase can be
    performed once and used for all queries. In
    the second phase, the actual sampling (\emph{SA}) of 10k trajectories (per
    object) is performed. The exact approach is denoted as
    \emph{EX}. In our default setting during efficiency analysis we
    set the number of objects $|\mathcal{D}|=10k$, the number of states $N = |\mathcal{S}| = 100$k,
    average branching factor of the synthetic graph $b=6$, probability threshold $\tau=0$ and the length of the query interval $|T| = 10$.
    These parameters lead to a total of $110k$ observations (11 per object)
    and $100k$ diamonds for the UST-index.

\textbf{Varying $N$.}
\begin{figure}[t] \centering
\includegraphics[width=\columnwidth]{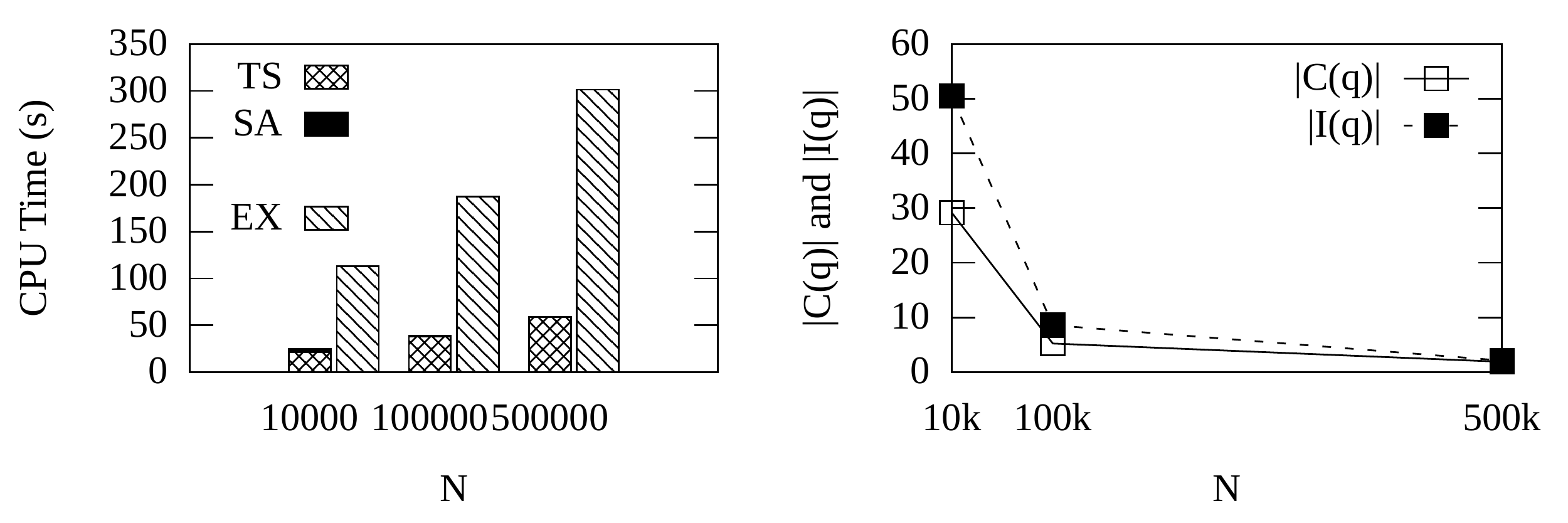}
          \vspace{-2.4em}
          \caption{Varying the Number of States}
          \label{fig:num_states}
            \vspace{-0.1cm}
        \end{figure}
In the first experiment (Figure \ref{fig:num_states}) we
investigate the effect of an increasing state space size $N$,
while keeping a constant average branching factor of network
nodes. This effect corresponds to expanding the underlying state space,
e.g., from a single country to a whole continent.
In Figure \ref{fig:num_states} (left) we can see that increasing $N$
leads to a sublinear increase in the run-time of both the sampling
approaches and the exact solution. This effect can be mostly
explained by two aspects. First, the size of the a-priori model
increases linearly with $N$, since the number of non-zero
elements of the sparse matrix $M$ increases linearly with $N$. This
leads to an increase of the time complexity of matrix operations. At the same
time, the number of candidates $|C(t)|$ and influence objects $I(t)$ decreases
significantly as seen in Figure \ref{fig:num_states} (right) because the degree
of intersection between objects decreases with a higher number of states, making
pruning more effective. The actual sampling cost $SA$, which is too small to be
noticeable in Figure \ref{fig:num_states} (left) decreases from $4$s for 10k states to 0.7s
for 500k states due to the smaller number of candidates and
influence objects.

\textbf{Varying $b$. }
   \begin{figure}[t] \centering
\includegraphics[width=\columnwidth]{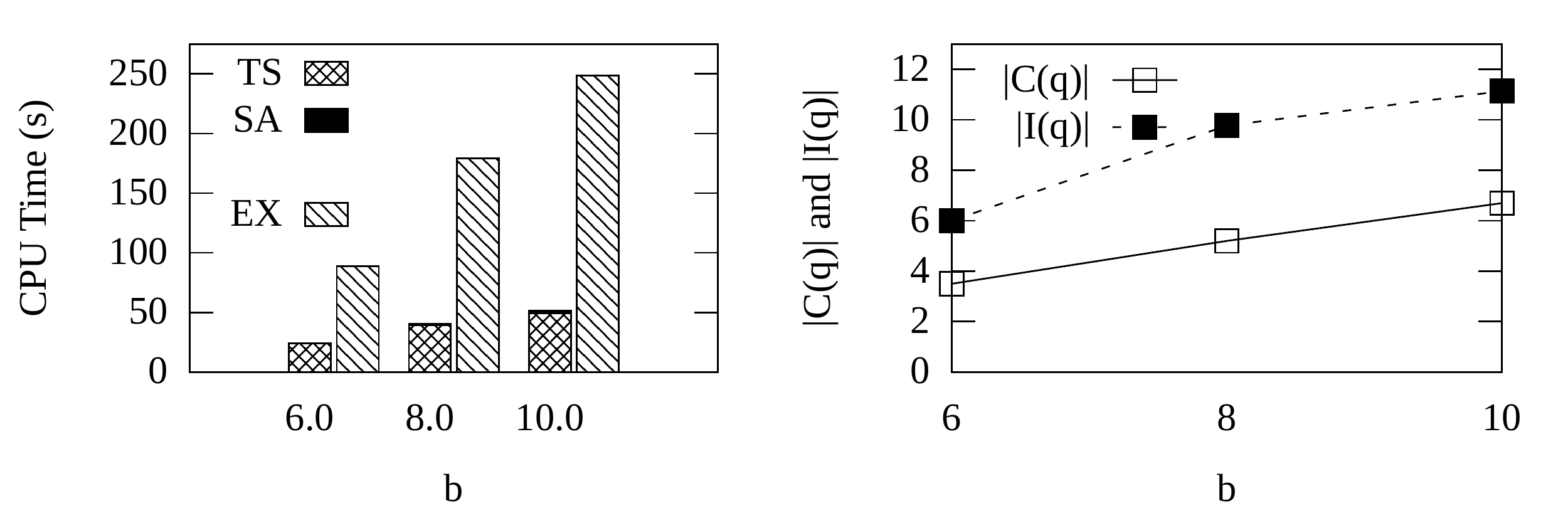}
          \vspace{-2.4em}
          \caption{Varying the Branching Factor}
          \label{fig:branching_factor}
            \vspace{-0.1cm}
        \end{figure}
Figure \ref{fig:branching_factor} evaluates the branching factor
$b$, i.e., the average degree of each network node. As expected,
Figure \ref{fig:branching_factor} (left) shows that an increasing
branching factor yields a higher run-time of all approaches due to
a higher number of non-zero values in vectors and matrices, making
computations more costly. Furthermore, in our setting, a larger
branching factor also increases the number of influence objects,
as shown in Figure \ref{fig:branching_factor} (right).

\textbf{Varying $|\db|$.}
   \begin{figure}[t] \centering
\includegraphics[width=\columnwidth]{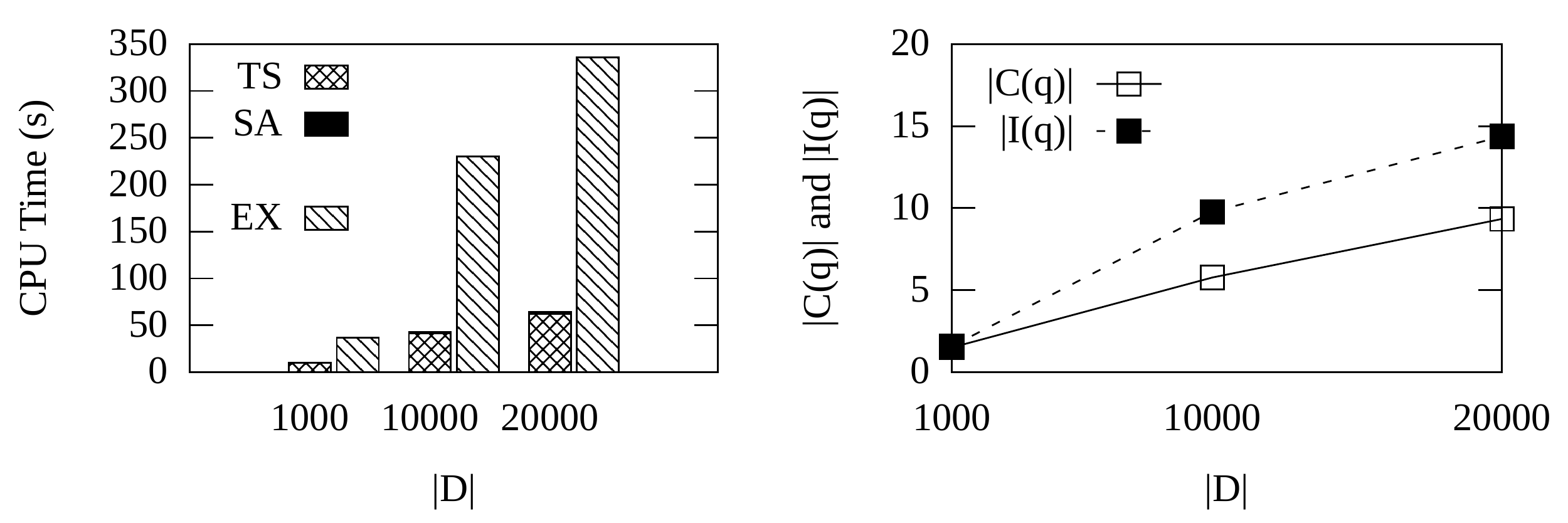}
\vspace{-2.4em}
          \caption{Varying the Number of Objects}
          \label{fig:num_objects}
            \vspace{-0.1cm}
        \end{figure}
The number of objects (Figure \ref{fig:num_objects}) leads to a decreasing
performance as well. The more objects stored in a database with the same
underlying motion model, the more candidates and influence objects are found
during the filter step. This leads to an increasing number of probability
calculations during refinement, and hence a higher query cost.

    \begin{figure}[t] \centering
        \subfigure{
            \includegraphics[width=0.45\columnwidth]{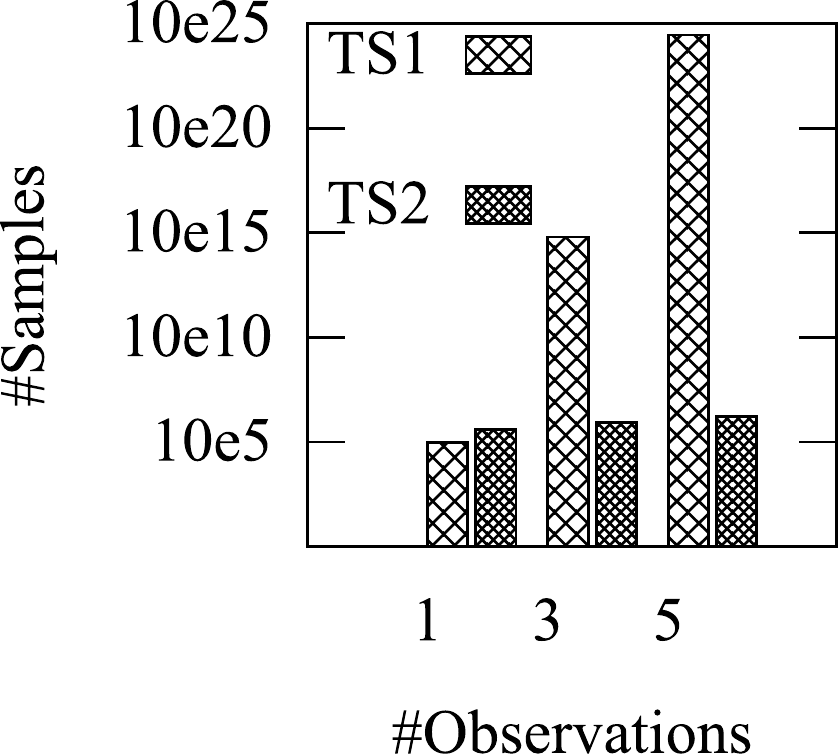}
            %\label{fig:trivial_sampling_experiment}
            }
        \subfigure{
           \centering
           \includegraphics[width=0.45\columnwidth]{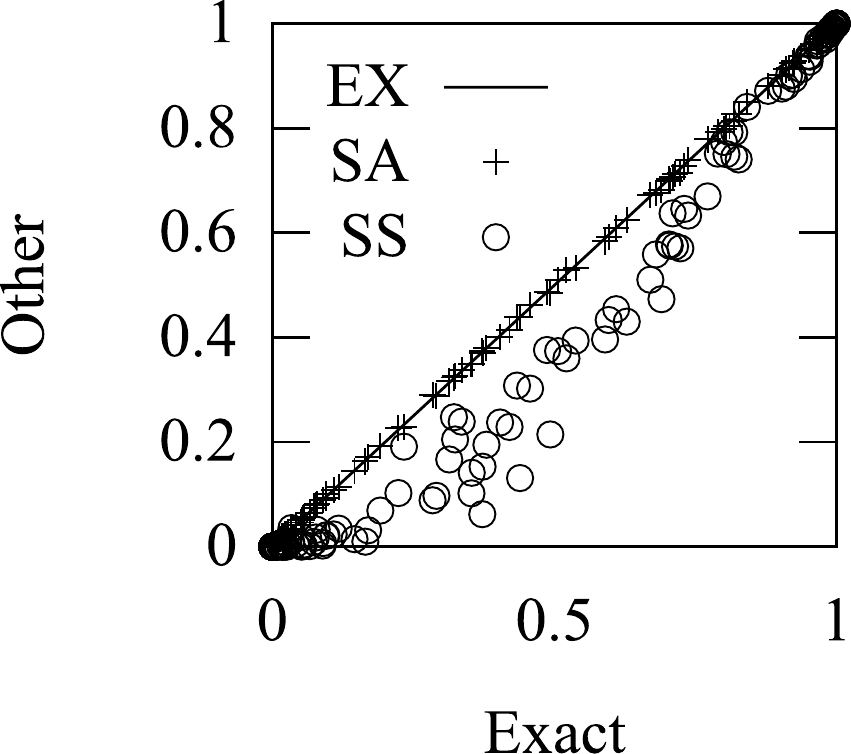}
           %\label{fig:snapshot_continuous}
           }\vspace{-0.2cm}
           \caption{Efficiency and Effectiveness of Sampling}
           \label{fig:effeff}
             \vspace{-0.3cm}
  \end{figure}

  \textbf{Sampling Efficiency.} In the next experiment we evaluate the overhead
 of the traditional sampling approach (using the a-priori Markov model only) compared to the approach
 presented in Section \ref{sec:sampling} which uses the a-posteriori model. The first,
 traditional approach (TS1) discards any trajectory not visiting all
 observations. As discussed in Section \ref{sec:sampling:naive}, the expected
 number of attempts required to draw one sample that hits all observations,
 increases exponentially in the number of observations. This increase is shown in
 Figure \ref{fig:effeff} (left), where the expected number of
 samples is depicted with respect to the number of observations. This approach
 can be improved, by segment-wise sampling between observations. Once the first
 observation is hit, the corresponding trajectory is memorized, and further
 samples from the current observation are drawn until the next observation is
 hit. The number of trajectories required to be drawn in order to obtain one
 possible trajectory is linear to the number of observations when using this
 approach (TS2). We note in Figure \ref{fig:effeff} (left), that in either approach,
 100k samples are required even in the case of having only two observations. The reason is
 that by generation, trajectories follow a near-shortest path, which is a highly
 unlikely scenario using the a-priori Markov model. Using the approach presented
 in Section \ref{sec:sampling}, the number of trajectories that need to
 be sampled, in order to obtain a trajectory that hits all observations, is always
 one.

 \textbf{Sampling Precision and Effectiveness.}
 Next, we evaluate the precision of our approximate P$\forall$NN query and an
 aspect of a competitor approach proposed in \cite{XuGuCheQiaoYu13}. The latter
 approach has been tailored for \emph{reverse} NN queries, but can easily be
 adapted to NN query processing. Essentially, this approach performs a snapshot
 query, i.e., $P\forall NN(o,q,\DB,t)$ for each $t \in T$. $P\forall
 NN(o,q,\DB,T)$ is then estimated by $\prod_{t\in T} P\forall NN(o,q,\DB,t)$.
 The scatterplot in Figure \ref{fig:effeff}(right) illustrates the results
of a series of $\forall NN$ queries ($v=0.2$, $|T|=5$).

At each experiment, we estimate probabilities by our sampling
approach (SA) (Section \ref{sec:sampling}) and by the adapted approach of
\cite{XuGuCheQiaoYu13} (SS). We model each case as a (x,y) point, where x models
the actual and y the estimated probability. For the exact method (EX) (Section
\ref{subsec:forallnnq}) the results always lie on the diagonal identity function
depicted by a straight line, showing that our sampling solution tightly approximates the results of the exact P$\forall$NN query. Concerning the
snapshot approach, a strong bias towards underestimating probabilities can be
observed. This bias is a result of treating points of time mutually independent.
In reality, the position at time $t$ must be in vicinity of the position at time
$t-1$, due to maximum speed constraints. This positive correlation in space
directly leads to a nearest neighbor correlation: If $o$ is close to $q$ at time
$t-1$, then $o$ is likely close to $q$ at time $t$. And clearly, if $o$ is more likely to be close to $q$ at time $t$, then $o$ is
more likely to be the NN of $q$ at time $t$. This correlation is ignored by
snapshot approaches.

The number of samples required to obtain an accurate approximation of the
probability of a binomial distributed random event such as the event that $o$ is
the NN of $q$ for each time $t \in T$ has been studied extensively in statistics
\cite{Hoe63}. Thus the required number of samples is not explicitly evaluated
here.

 \textbf{Real Dataset.}
 \begin{figure}[t] \centering
\includegraphics[width=\columnwidth]{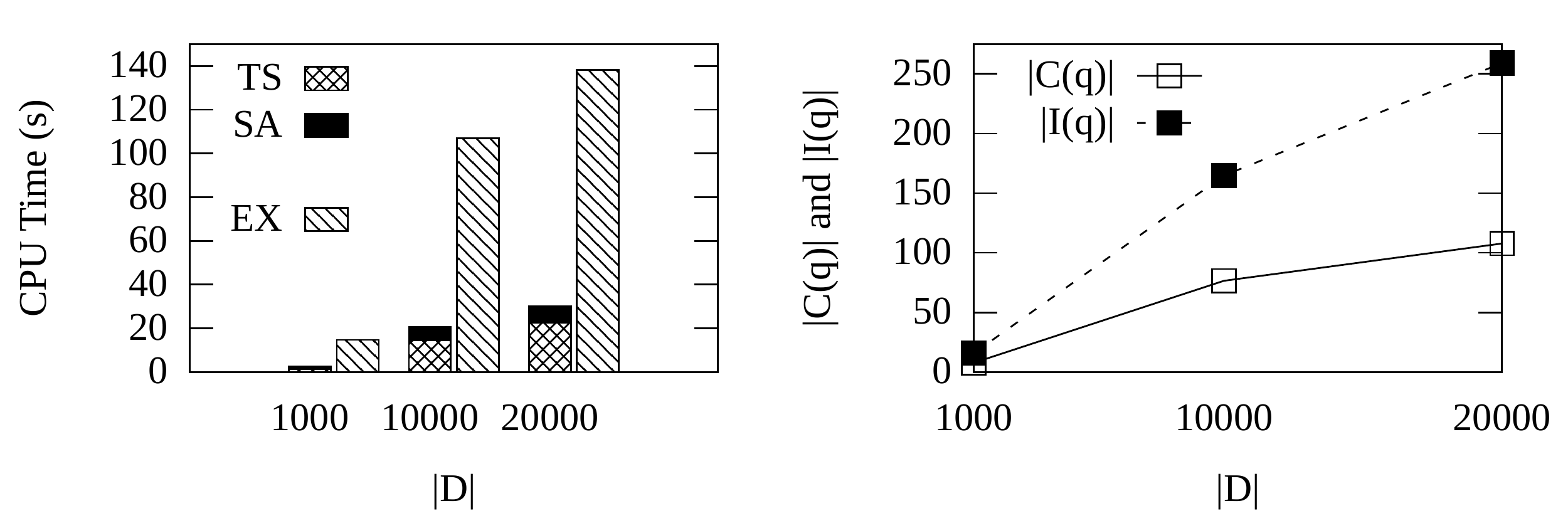}
  \vspace{-2em}
  \caption{Realdata: Varying the number of objects}
  \label{fig:num_objects_real}
\vspace{-0.1cm}
\end{figure}
We conducted additional experiments to evaluate P$\forall$NN queries on
the taxi dataset (Figure \ref{fig:num_objects_real}). Since the underlying
state space consisting of 3000 states is very small, we set $i=5$ and $v=0.6$ in order to
prevent uncertainty regions of objects to cover the whole network.
Based on this dataset, we ran an experiment varying the number of
objects between $1000$ and $20000$. The small size of the state
space leads to a higher objects density, leading to a larger
number of candidates and influence objects than the corresponding experiment on
the artificial dataset. On the other hand, a smaller state space and the lower level of
uncertainty decreases the complexity of matrix and vector operations.
Additionally, a higher number of candidates and influence objects
also decrease the probability that an object is a result of the
P$\forall$NN query, i.e., candidates are often pruned after
considering only a small number of influence objects. As a result,
the runtime cost on the real dataset is generally lower than on the
synthetic dataset.

\begin{figure}[t] \centering
\includegraphics[width=0.6\columnwidth]{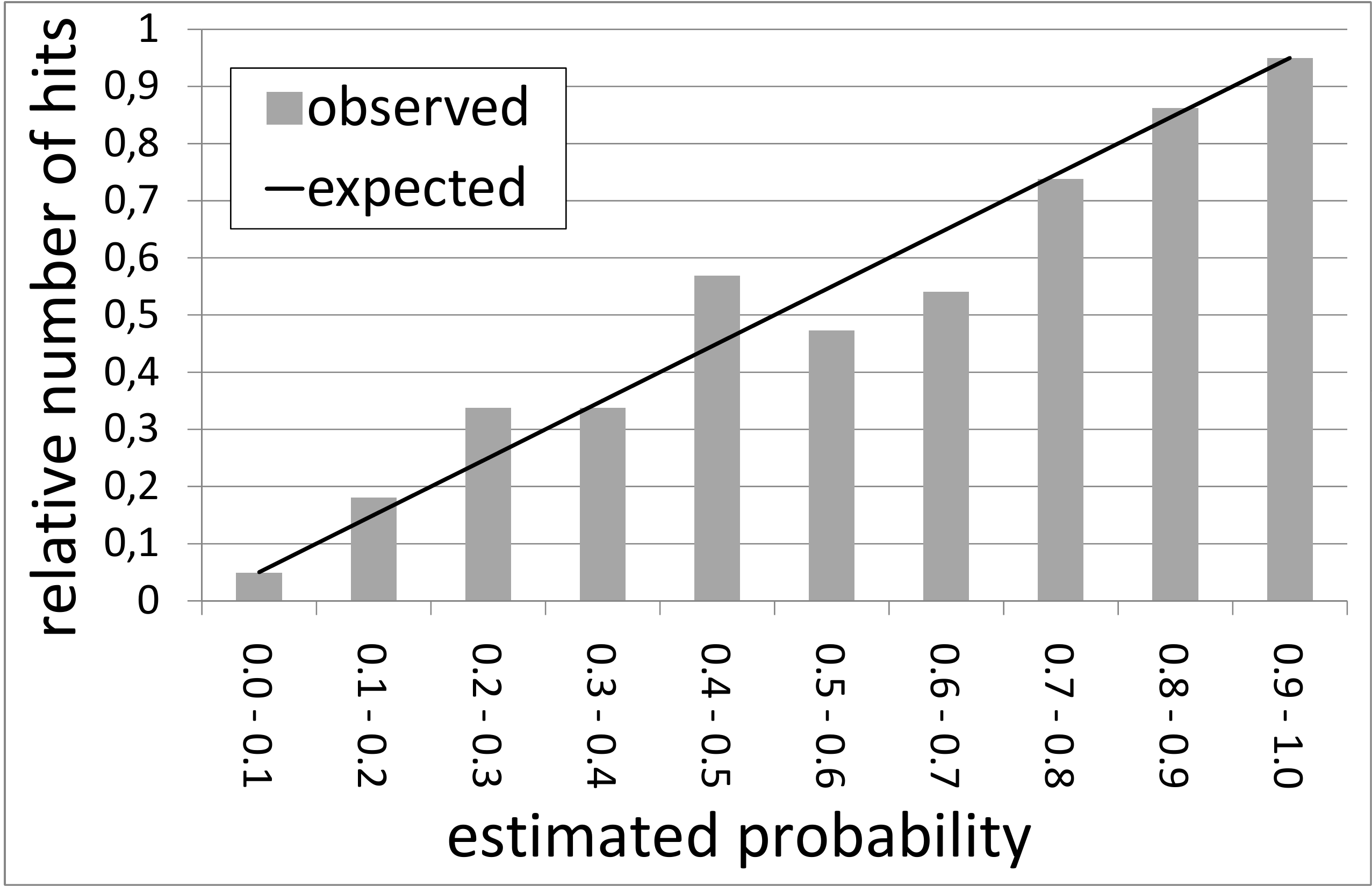}
  \caption{Quality}
  \label{fig:quality}
  \vspace{-0.1cm}
\end{figure}

 \textbf{Model Quality.} In the next experiment, we evaluate the quality of the
result probabilities of P$\forall$NN queries using the Markov model.
Therefore, our first aim is to generate interesting trajectories, as
we expect them to appear in real applications. The trajectories must
not follow an (unrealistic) Markov random walk (corresponding to our
a-priori model). At the same time, trajectories should not move on
the perfect shortest paths, as then only one trajectory may be
possible between two observations, leading to a perfect a-posteriori
model. Thus we generate ``\emph{near-shortest paths}'' from
shortest paths, by adding a wrong turn in every ten states. This
deviation simulates random errors of moving objects, e.g., due to
human error. Observations are taken from these trajectories at every
$i=10$ states. We perform 2500 P$\forall$NN queries using the
sampling approach. For each result object $o$ having a non-zero,
non-one result probability, a tuple is generated containing the
computed probability, as well as an indicator variable that is one
if $o$ is a \emph{true} $\forall$NN of $q$ and zero otherwise. This
ground truth is obtained by utilizing the full trajectory
information of all generated objects. The resulting tuples are
grouped by probability, with the fraction of $\forall$NNs being
aggregated. The result of this experiment is shown in Figure
\ref{fig:quality}. Here, result probabilities are grouped into ten
intervals. For each probability interval, the expected number of
hits, assuming that the estimated probability is correct is depicted
by a straight line. The observed number of hits, evaluated on the
ground truth is shown by bars. For each bucket, a small deviation
can be explained by random deviation, due to the fact, that fewer
tuples were derived in the probability
interval $[0.2,0.8]$. 

This result implies that our computed a-posteriori Markov model,
which is adapted to existing observations, is able to effectively
model the generated uncertain trajectories. This is notable, because
the generated trajectories do not, as the a-priori model assumes,
perform a weighted random walk. In fact, these trajectories are
generated such that they are close to the shortest path.

 \subsection{Continuous Queries}

In our experimental evaluation on continuous queries we compare the
runtime cost and the size of the (unprocessed) result set for various sizes of
the database and values of the threshold $\tau$.
\begin{figure}[t] \centering
\includegraphics[width=\columnwidth]{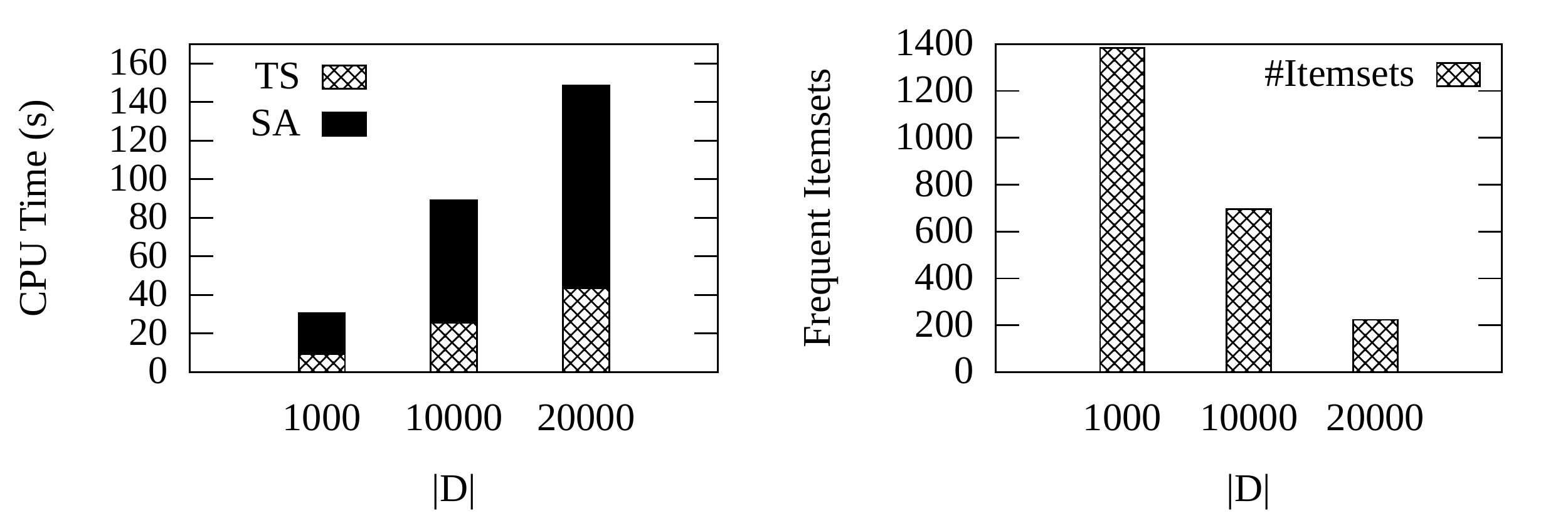}
  \vspace{-2em}
  \caption{Continuous Queries: Varying the number of objects}
  \label{fig:continuous_num_objects}
\vspace{-0.1cm}
\end{figure}

Increasing the number of objects stored in the database leads to an
increase in the time $TS$ to compute the a-posteriori Markov model
for each object (cf. Figure \ref{fig:continuous_num_objects} (left)).
This result is equivalent to the result for P$\forall$NN queries,
since a-posteriori models have to be computed for either query
semantics. However, the time required to obtain a sufficient number
of samples (\emph{SA}) is much higher, since probabilities have to
be estimated for a number of sets of time intervals, rather than for
the single interval $T$. This increase in run-time is alleviated by
the effect that the number of candidates obtained in the candidate
generation step of our Apriori-like algorithm decreases (Figure
\ref{fig:continuous_num_objects} (right)). This effect follows from
the fact that more objects lead to more pruners, leading to smaller
probabilities of intervals, leading to fewer candidates.
\begin{figure}[t] \centering
\includegraphics[width=\columnwidth]{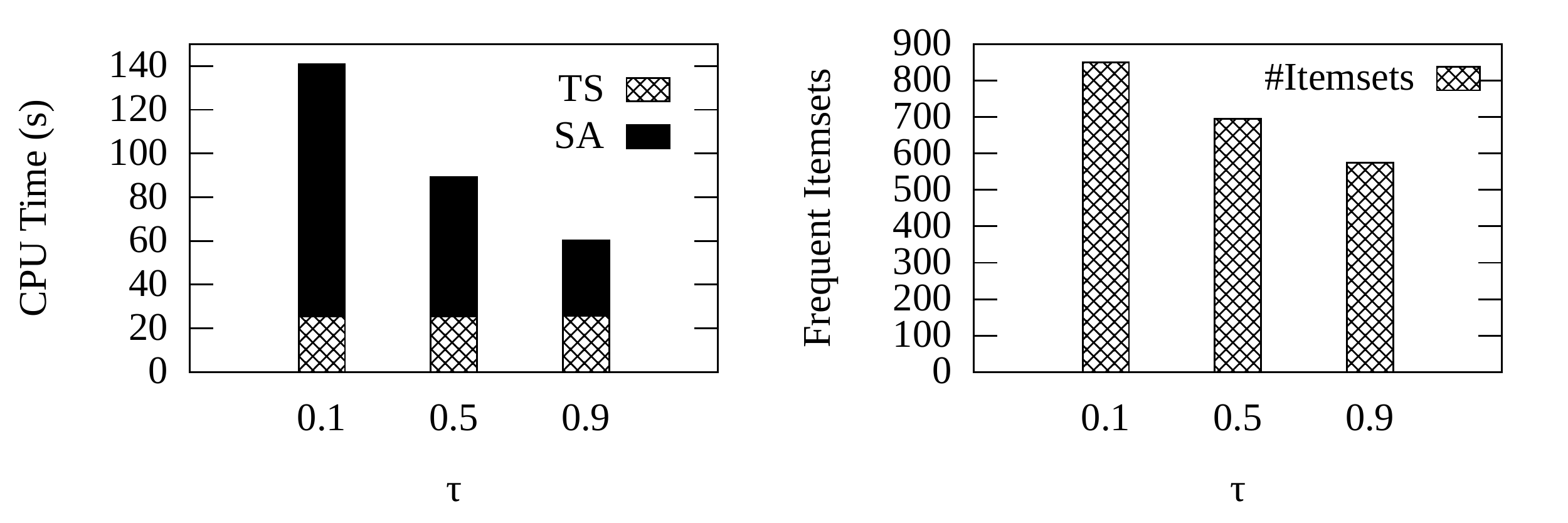}
  \vspace{-2em}
  \caption{Continuous Queries: Varying $\tau$}
  \label{fig:continuous_tau}
    \vspace{-0.1cm}
\end{figure}

The results of varying $\tau$ can be found in Figure
\ref{fig:continuous_tau}. Clearly an increasing probability threshold decreases
the average size of the result (Figure
\ref{fig:continuous_tau} (right)). Consequently, at the same time, the
computational complexity of the query decreases as fewer candidates are generated. Figure
\ref{fig:continuous_tau}(left) shows that the run-time of the sampling approach
becomes very large for low values of $\tau$, since samples have to be evaluated
for each generated candidate set. Similar to the Apriori-algorithm, the number
of such candidates grows exponentially with $T$, if $\tau$ is small.

\section{$k$-Nearest-Neighbor Queries}\label{sec:knn}
To answer P$\exists k$NN queries, P$\forall k$NN queries and PC$k$NN
queries approximately in the case of $k>1$, we can again utilize
the model adaptation and sampling technique presented in Section
\ref{sec:sampling}. Therefore, possible worlds are sampled using the
a-posteriori models of all objects, given their observations. On
each such (certain) world an existing solution for $k$NN search on
certain trajectories
(e.g. \cite{FreGraPelThe07,GueBehXu10,IweSamSmi03,TaoPapShe02}) is applied. The
results of these deterministic queries can again be used to estimate the distribution of the probabilistic result.

Here, we briefly discuss the complexity of computing exact
results of P$\exists k$NN queries, P$\forall k$NN queries and PC$k$NN
queries. A formal definition of these queries, as well as a more
detailed discussion can be found in our technical report
(\cite{NieZueEmrRenetal13}).

The P$\forall k$NN query is NP-hard in $k$. The proof of this
statement can be found in our technical report \cite{NieZueEmrRenetal13}. To
summarize, the proof shows that for the special case where
$|\DB|=k+1$ the problem can be reduced to an P$\exists$NN query
which has been shown to be
  NP-hard in $|\DB|$ in Section \ref{sec:existsnn}. We then extend this proof for arbitrary
  sizes of the database $\DB$, and show that, as long as $|\DB| > k$,
  the run-time of a
  P$\forall k$NN query is at least exponential in $k$. The NP-hardness of the
  P$\exists k$NN query is shown in a similar fashion in
  \cite{NieZueEmrRenetal13}. Finally, the continuous $\forall k$NN which is
  based on the $\forall k$NN query, is also shown to be hard in $k$.

\section{Conclusions}\label{sec:conclusions}
  In this paper, we addressed the problem of answering NN queries in uncertain
  spatio-temporal databases under temporal dependencies. We proposed three
  different semantics of NN queries: P$\forall$NN queries, P$\exists$NN queries
  and PCNN queries.
  We have shown that the P$\forall$NN query can be solved in polynomial time,
  while the P$\exists$NN query is NP-hard.
  These results provide insights about the complexity of NN search over
  uncertain data in general since the Markov chain model is one of the simplest
  models that consider temporal dependencies. More complex models are expected
  to be at least as hard. To mitigate the problems of computational complexity,
  we propose a sampling-based approach based on Bayesian inference. For 
  the query problems that can be solved in PTIME, we presented exact
  query evaluation algorithms. Specifically, for the PCNN query we proposed to
  reduce the cardinality of the result set by means of an Apriori pattern mining
  approach. To cope with large trajectory databases, we introduced a pruning
  strategy to speed-up PNN queries exploiting the UST tree, an index for
  uncertain trajectory data. The experimental evaluation shows that our adapted
  a-posteriori model allows to effectively and efficiently answer probabilistic
  NN queries despite the strong a-priori Markov assumption.

\vspace{1cm}

\begin{scriptsize}
\bibliographystyle{IEEEtran}
\bibliography{abbrev,literature}

\begin{thebibliography}{10}
\providecommand{\url}[1]{#1}
\csname url@samestyle\endcsname
\providecommand{\newblock}{\relax}
\providecommand{\bibinfo}[2]{#2}
\providecommand{\BIBentrySTDinterwordspacing}{\spaceskip=0pt\relax}
\providecommand{\BIBentryALTinterwordstretchfactor}{4}
\providecommand{\BIBentryALTinterwordspacing}{\spaceskip=\fontdimen2\font plus
\BIBentryALTinterwordstretchfactor\fontdimen3\font minus
  \fontdimen4\font\relax}
\providecommand{\BIBforeignlanguage}[2]{{%
\expandafter\ifx\csname l@#1\endcsname\relax
\typeout{** WARNING: IEEEtran.bst: No hyphenation pattern has been}%
\typeout{** loaded for the language `#1'. Using the pattern for}%
\typeout{** the default language instead.}%
\else
\language=\csname l@#1\endcsname
\fi
#2}}
\providecommand{\BIBdecl}{\relax}
\BIBdecl

\bibitem{FreGraPelThe07}
E.~Frentzos, K.~Gratsias, N.~Pelekis, and Y.~Theodoridis, ``Algorithms for
  nearest neighbor search on moving object trajectories,''
  \emph{Geoinformatica}, vol.~11, no.~2, pp. 159--193, 2007.

\bibitem{GueBehXu10}
R.~H. G{\"u}ting, T.~Behr, and J.~Xu, ``Efficient {\it k}-nearest neighbor
  search on moving object trajectories,'' \emph{VLDB J.}, vol.~19, no.~5, pp.
  687--714, 2010.

\bibitem{IweSamSmi03}
G.~S. Iwerks, H.~Samet, and K.~Smith, ``Continuous k-nearest neighbor queries
  for continuously moving points with updates,'' in \emph{Proc.\ VLDB}.\hskip
  1em plus 0.5em minus 0.4em\relax VLDB Endowment, 2003, pp. 512--523.

\bibitem{TaoPapShe02}
Y.~Tao, D.~Papadias, and Q.~Shen, ``Continuous nearest neighbor search,'' in
  \emph{VLDB}, 2002, pp. 287--298.

\bibitem{EmrKriMamRenZue12}
T.~Emrich, H.-P. Kriegel, N.~Mamoulis, M.~Renz, and A.~Züfle, ``Indexing
  uncertain spatio-temporal data,'' in \emph{Proc.\ CIKM}, 2012, pp. 395--404.

\bibitem{TaoFalPapLiu04}
Y.~Tao, C.~Faloutsos, D.~Papadias, and B.~Liu, ``Prediction and indexing of
  moving objects with unknown motion patterns,'' in \emph{SIGMOD Conference},
  2004, pp. 611--622.

\bibitem{YuPuKou05}
X.~Yu, K.~Q. Pu, and N.~Koudas, ``Monitoring k-nearest neighbor queries over
  moving objects,'' in \emph{ICDE}, 2005, pp. 631--642.

\bibitem{YioMokAre05}
X.~Xiong, M.~F. Mokbel, and W.~G. Aref, ``Sea-cnn: Scalable processing of
  continuous k-nearest neighbor queries in spatio-temporal databases,'' in
  \emph{ICDE}, 2005, pp. 643--654.

\bibitem{MokSu04}
H.~Mokhtar and J.~Su, ``Universal trajectory queries for moving object
  databases,'' in \emph{Proc.\ MDM}, 2004, pp. 133--144.

\bibitem{TraWolHinCha04}
G.~Trajcevski, O.~Wolfson, K.~Hinrichs, and S.~Chamberlain, ``Managing
  uncertainty in moving objects databases,'' \emph{ACM Trans. Database Syst.},
  vol.~29, no.~3, pp. 463--507, 2004.

\bibitem{TraTamDinSchCru09}
G.~Trajcevski, R.~Tamassia, H.~Ding, P.~Scheuermann, and I.~F. Cruz,
  ``Continuous probabilistic nearest-neighbor queries for uncertain
  trajectories,'' in \emph{Proc.\ EDBT}, 2009, pp. 874--885.

\bibitem{EmrEtAl11}
T.~Emrich, H.-P. Kriegel, N.~Mamoulis, M.~Renz, and A.~Züfle, ``Querying
  uncertain spatio-temporal data,'' in \emph{Proc.\ ICDE}, 2012, pp. 354--365.

\bibitem{TraChoWolYeLi10}
G.~Trajcevski, A.~N. Choudhary, O.~Wolfson, L.~Ye, and G.~Li, ``Uncertain range
  queries for necklaces,'' in \emph{Proc.\ MDM}, 2010, pp. 199--208.

\bibitem{CheKalPra04}
R.~Cheng, D.~Kalashnikov, and S.~Prabhakar, ``Querying imprecise data in moving
  object environments,'' in \emph{IEEE TKDE}, vol.~16, no.~9, 2004, pp.
  1112--1127.

\bibitem{TraTamDinSchetal09}
G.~Trajcevski, R.~Tamassia, H.~Ding, P.~Scheuermann, and I.~F. Cruz,
  ``Continuous probabilistic nearest-neighbor queries for uncertain
  trajectories,'' in \emph{Proc.\ EDBT}, 2009, pp. 874--885.

\bibitem{QiaTanJinLonDaiKuCha10}
S.~Qiao, C.~Tang, H.~Jin, T.~Long, S.~Dai, Y.~Ku, and M.~Chau, ``Putmode:
  prediction of uncertain trajectories in moving objects databases,''
  \emph{Appl. Intell.}, vol.~33, no.~3, pp. 370--386, 2010.

\bibitem{ReLetBalSuc08}
C.~R\'{e}, J.~Letchner, M.~Balazinksa, and D.~Suciu, ``Event queries on
  correlated probabilistic streams,'' in \emph{Proc.\ SIGMOD}, 2008, pp.
  715--728.

\bibitem{XuGuCheQiaoYu13}
C.~Xu, Y.~Gu, L.~Chen, J.~Qiao, and G.~Yu, ``Interval reverse nearest neighbor
  queries on uncertain data with markov correlations,'' in \emph{Proc.\ ICDE},
  2013.

\bibitem{KolGunTso99}
G.~Kollios, D.~Gunopulos, and V.~Tsotras, ``Nearest neighbor queries in a
  mobile environment,'' in \emph{Spatio-Temporal Database Management}.\hskip
  1em plus 0.5em minus 0.4em\relax Springer, 1999, pp. 119--134.

\bibitem{PraWolChaDao97}
A.~Prasad~Sistla, O.~Wolfson, S.~Chamberlain, and S.~Dao, ``Modeling and
  querying moving objects,'' in \emph{Proc.\ ICDE}.\hskip 1em plus 0.5em minus
  0.4em\relax IEEE, 1997, pp. 422--432.

\bibitem{HuaLiaLee08}
Y.-K. Huang, S.-J. Liao, and C.~Lee, ``Efficient continuous k-nearest neighbor
  query processing over moving objects with uncertain speed and direction,'' in
  \emph{SSDBM}, 2008, pp. 549--557.

\bibitem{LiLiShuFan11}
G.~Li, Y.~Li, L.~Shu, and P.~Fan, ``Cknn query processing over moving objects
  with uncertain speeds in road networks,'' in \emph{APWeb}, 2011, pp. 65--76.

\bibitem{TraTamCruSchHarZam11}
G.~Trajcevski, R.~Tamassia, I.~F. Cruz, P.~Scheuermann, D.~Hartglass, and
  C.~Zamierowski, ``Ranking continuous nearest neighbors for uncertain
  trajectories,'' \emph{VLDB J.}, vol.~20, no.~5, pp. 767--791, 2011.

\bibitem{AgrSri94}
R.~Agrawal and R.~Srikant, ``Fast algorithms for mining association rules,'' in
  \emph{Proc.\ VLDB}, 1994, pp. 487--499.

\bibitem{JamXuWuetal08}
R.~Jampani, F.~Xu, M.~Wu, L.~L. Perez, C.~M. Jermaine, and P.~J. Haas, ``Mcdb:
  a monte carlo approach to managing uncertain data,'' in \emph{Proc.\ SIGMOD},
  2008, pp. 687--700.

\bibitem{NieZueEmrRenetal13}
J.~Niedermayer, A.~Züfle, T.~Emrich, M.~Renz, N.~Mamoulis, L.~Chen, and H.-P.
  Kriegel, ``Probabilistic nearest neighbor queries on uncertain moving object
  trajectories (technical report),'' 2013,
  \url{http://www.dbs.ifi.lmu.de/Publikationen/Papers/TR_PNN.pdf}.

\bibitem{Hoe63}
W.~Hoeffding, ``Probability inequalities for sums of bounded random
  variables,'' \emph{Journal of the American Statistical Association}, pp.
  13--30, 1963.

\bibitem{BecKriSchSee90}
N.~Beckmann, H.-P. Kriegel, R.~Schneider, and B.~Seeger, ``{The R*-Tree}: An
  efficient and robust access method for points and rectangles,'' in
  \emph{Proc.\ SIGMOD}, 1990.

\bibitem{ProjectPage13}
\BIBentryALTinterwordspacing
The ust project page. [Online]. Available:
  \url{http://www.dbs.ifi.lmu.de/cms/Publications/UncertainSpatioTemporal}
\BIBentrySTDinterwordspacing

\bibitem{YuaZheXieSun11}
J.~Yuan, Y.~Zheng, X.~Xie, and G.~Sun, ``Driving with knowledge from the
  physical world,'' in \emph{Proc.\ KDD}, 2011, pp. 316--324.

\bibitem{CheSheZhou11}
Z.~Chen, H.~T. Shen, and X.~Zhou, ``Discovering popular routes from
  trajectories,'' in \emph{Proc.\ ICDE}, 2011, pp. 900--911.

\end{thebibliography}
\end{scriptsize}   

\end{document}